\PassOptionsToPackage{table}{xcolor}
\documentclass[a4paper,USenglish,cleveref,autoref,thm-restate]{lipics-v2021}

\usepackage{mathtools}      %
\usepackage{dsfont}
\usepackage{xfrac}
\usepackage{pifont}
\usepackage{ifthen}
\usepackage{tikz}
\usetikzlibrary{math,calc,positioning,arrows,arrows.meta,shapes.geometric,fit,backgrounds}
\usepackage{placeins}       %

\usepackage[disable]{todonotes}
\newcommand{\twrchanged}[1]{{{#1}}}

\usepackage{fontawesome5}   %

\newcommand{\leanref}[2]{}

\newcommand{\co}{\mathbb{C}}
\newcommand{\z}{\mathbb{Z}}
\newcommand{\natur}{\mathbb{N}}
\newcommand{\seq}{\text{Seq}}
\newcommand{\Rot}{R_{\infty}}
\newcommand{\ctl}{\mathrm{ctl}}
\let\bra\relax
\let\ket\relax
\DeclarePairedDelimiter\bra{\langle}{|}
\DeclarePairedDelimiter\ket{|}{\rangle}

\providecommand{\gen}[1]{\ensuremath{\left\langle #1 \right\rangle}}
\newcommand{\id}{I}
\newcommand{\beforeq}{\preccurlyeq}
\newcommand{\highlim}{B_{\mathrm{high}}}
\newcommand{\rootlim}{B_{\mathrm{root}}}
\newcommand{\Yes}{\ding{51}}
\newcommand{\No}{\ding{55}}

\tikzstyle{leaf}=[draw, rectangle,minimum size=4.mm, inner sep=3pt]
\tikzstyle{var}=[circle,draw=black!70,solid,thick,minimum size=6mm]
\tikzstyle{bdd}=[regular polygon, regular polygon sides=3, draw=black!70,solid,thick,inner sep=0.5mm]
\tikzstyle{n}=[->,loosely dashed,thick]
\tikzstyle{p}=[->,solid,thick]
\tikzstyle{b}=[->,densely dashdotted,ultra thick]
\tikzstyle{e0}[0]=[dashed,thick,bend right=#1,->,thick]
\tikzstyle{e1}[0]=[solid, bend left =#1,->,thick]
\tikzstyle{lbl}=[draw,fill=white,inner sep=2pt, minimum size=0cm,line width=.5pt]

\def\rotateclockwise#1{
  \newdimen\xrw
  \pgfextractx{\xrw}{#1}
  \newdimen\yrw
  \pgfextracty{\yrw}{#1}
  \pgfpoint{\yrw}{-\xrw}
}

\def\rotatecounterclockwise#1{
  \newdimen\xrcw
  \pgfextractx{\xrcw}{#1}
  \newdimen\yrcw
  \pgfextracty{\yrcw}{#1}
  \pgfpoint{-\yrcw}{\xrcw}
}

\def\outsidespacerpgfclockwise#1#2#3{
  \pgfpointscale{#3}{
    \rotateclockwise{
      \pgfpointnormalised{
        \pgfpointdiff{#1}{#2}}}}
}

\def\outsidespacerpgfcounterclockwise#1#2#3{
  \pgfpointscale{#3}{
    \rotatecounterclockwise{
      \pgfpointnormalised{
        \pgfpointdiff{#1}{#2}}}}
}

\def\outsidepgfclockwise#1#2#3{
  \pgfpointadd{#2}{\outsidespacerpgfclockwise{#1}{#2}{#3}}
}

\def\outsidepgfcounterclockwise#1#2#3{
  \pgfpointadd{#2}{\outsidespacerpgfcounterclockwise{#1}{#2}{#3}}
}

\def\outside#1#2#3{
  ($ (#2) ! #3 ! -90 : (#1) $)
}

\def\cornerpgf#1#2#3#4{
  \pgfextra{
    \pgfmathanglebetweenpoints{#2}{\outsidepgfcounterclockwise{#1}{#2}{#4}}
    \let\anglea\pgfmathresult
    \let\startangle\pgfmathresult

    \pgfmathanglebetweenpoints{#2}{\outsidepgfclockwise{#3}{#2}{#4}}
    \pgfmathparse{\pgfmathresult - \anglea}
    \pgfmathroundto{\pgfmathresult}
    \let\arcangle\pgfmathresult
    \ifthenelse{180=\arcangle \or 180<\arcangle}{
      \pgfmathparse{-360 + \arcangle}}{
      \pgfmathparse{\arcangle}}
    \let\deltaangle\pgfmathresult

    \newdimen\x
    \pgfextractx{\x}{\outsidepgfcounterclockwise{#1}{#2}{#4}}
    \newdimen\y
    \pgfextracty{\y}{\outsidepgfcounterclockwise{#1}{#2}{#4}}
  }
  -- (\x,\y) arc [start angle=\startangle, delta angle=\deltaangle, radius=#4]
}

\def\corner#1#2#3#4{
  \cornerpgf{\pgfpointanchor{#1}{center}}{\pgfpointanchor{#2}{center}}{\pgfpointanchor{#3}{center}}{#4}
}

\def\hedgeiii#1#2#3#4{
  \outside{#1}{#2}{#4} \corner{#1}{#2}{#3}{#4} \corner{#2}{#3}{#1}{#4} \corner{#3}{#1}{#2}{#4} -- cycle
}

\def\hedgem#1#2#3#4{
  
  \outside{#1}{#2}{#4}
  \pgfextra{
    \def\hgnodea{#1}
    \def\hgnodeb{#2}
  }
  foreach \c in {#3} {
    \corner{\hgnodea}{\hgnodeb}{\c}{#4}
    \pgfextra{
      \global\let\hgnodea\hgnodeb
      \global\let\hgnodeb\c
    }
  }
  \corner{\hgnodea}{\hgnodeb}{#1}{#4}
  \corner{\hgnodeb}{#1}{#2}{#4}
  -- cycle
}

\def\hedgeii#1#2#3{
  \hedgem{#1}{#2}{}{#3}
}

\title{Generalized LIMDDs: Succinctness and Canonicity for Decision Diagrams Modulo a Group}
\titlerunning{Generalized LIMDDs}

\author{Arend-Jan Quist}{Leiden Institute of Advanced Computer Science, Leiden University, Leiden, The Netherlands}{}{https://orcid.org/0000-0002-6501-2112}{}
\author{Alexis de Colnet}{Leiden Institute of Advanced Computer Science, Leiden University, Leiden, The Netherlands}{}{https://orcid.org/0000-0002-7517-6735}{}
\author{Thomas Reps}{Department of Computer Sciences, University of Wisconsin-Madison, Madison, Wisconsin, USA}{}{https://orcid.org/0000-0002-5676-9949}{}
\author{Alfons Laarman}{Leiden Institute of Advanced Computer Science, Leiden University, Leiden, The Netherlands}{}{https://orcid.org/0000-0002-2433-4174}{}

\authorrunning{A.-J. Quist, A. de Colnet, T. Reps, and A. Laarman}

\Copyright{Arend-Jan Quist, Alexis de Colnet, Thomas Reps, and Alfons Laarman}



\nolinenumbers

\begin{document}
\maketitle

\begin{abstract}
  A reduced ordered decision diagram is the minimal automaton of a
  function on words of fixed length: its nodes are the residuals, and reduction is
  the Myhill--Nerode quotient. We study what happens when that quotient is coarsened
  by a group. 
  Fix a group $G$ acting on residuals, merge two nodes when their
  residuals lie in one $G$-orbit, and record the group element on the edge. For
  functions $\{0,1\}^n\to\mathbb{C}$ and $G$ the Pauli group this is the Local
  Invertible Map Decision Diagram. We take $G$ from the two-parameter family
  generated by $C^{\leq k}R_q$, the phase rotation of order $q$ with up to $k$
  control qubits, with and without the bit flip $X$. We show that this gives exponential succinctness improvements compared to Pauli-LIMDD, and we determine the succinctness order of the family completely. The separating objects are hypergraph states, which can always be efficiently represented by some member of the family. We settle the tractability of five
  queries and eight transformations, which is invariant across the family, and shows the same behavior as Pauli-LIMDD. We give a five-rule reduction
  system whose normal forms are unique for every member of the family, and we show that this canonical form is computable in polynomial time in the size of the LIMDD.
  We show that, when coarsening beyond the (anti-)diagonal groups, the calculation of a minimal sized normal form turns out to be non-local and it might to rebuild the whole diagram. 
\end{abstract}

\clearpage

\section{Introduction}\label{sec:intro}

\twrchanged{
Ordered Binary Decision Diagrams~\cite{bryant86} (OBDDs) are a widely used data
structure for representing a pseudo-Boolean function
$f\colon\{0,1\}^n\to D$ in a compressed form.
An OBDD can be considered to be a minimal automaton that maps words of length $n$
to a value in $D$.
An OBDD can be represented as a directed acyclic graph (dag) that has a unique normal form, and hence
equality can be tested by a pointer comparison.
The \emph{residuals} of $f$ are the sub-functions (sub-dags) obtained by fixing a prefix.
The size is the number of distinct residuals, so lower bounds can be obtained via
counting arguments.
The normal form can be computed bottom-up, one layer at a time, because residual equality
at a node is decided from residual equality at its children.
}

\twrchanged{
Several generalizations of OBDDs have been introduced by putting a
\emph{label}, drawn from some label-domain $\mathcal{L}$ on each edge.
The semantics of a diagram with root edge $e$, called its \emph{state}
and denoted by $\ket e$, is a sequence defined as follows, where
$\ket{0}$ and $\ket{1}$ are the vectors $[1~0]^t$ and $[0~1]^t$,
respectively, and $\otimes$ denotes Kronecker product.
\begin{itemize}
    \item an edge $e$ into $v$ represents $\ket e=label(e) \cdot \ket v$;
    \item a node $v$ with $0$-edge $e_0$ and $1$-edge $e_1$ represents
        $\ket v=\ket0\otimes\ket{e_0}+\ket1\otimes\ket{e_1}$;
    \item the leaf represents the value $1$.
\end{itemize}
The state represented by a diagram is the state represented by its (labeled) root edge.

Two nodes can be merged whenever the states they represent agree
\emph{up to a label}, and labels that relate them (to some common
state) are recorded on the respective incoming edges.
Thus, the richer $\mathcal L$ is, the more nodes collapse, and
the smaller the diagram.
Writing $f,f'$ for the states of two nodes, $\mathcal{L}$-labeled decision diagrams
can be classified as follows:

\begin{center}\small
\begin{tabular}{lll}
  \textbf{merge two nodes when~~~~~~~} & \textbf{label set $\mathcal L$~~~~~~~} & \textbf{diagram}\\\hline
  $f=f'$
    & $\{1\}$
    & OBDDs~\cite{bryant86}\\
  $f=\lambda\cdot f'$
    & $\co\setminus\{0\}$
    & EVDDs~\cite{lai1994evbdd,qmdd}\\
  $f=\lambda\cdot P_1\otimes\dots\otimes P_n\cdot f'$
    & Pauli LIMs
    & Pauli-LIMDDs~\cite{limdd}\\
  $f=\lambda\cdot g\cdot f'$, $g\in\gen G$
    & $G$-LIMs
    & $\gen G$-LIMDDs (this paper)
\end{tabular}
\end{center}

\noindent
This paper takes on the general case in the last row where the label set is a group $\gen G$.
}
Two nodes now merge when their residuals lie in one $G$-orbit.
The edge records the group element that relates them.
Quotienting by orbit equality instead of equality can only shrink the diagram.
The price we pay is that a node now denotes an orbit, not a residual.
So the normal form needs a unique identifier per orbit, and a reduction rule
that computes it.

One instance is already understood.
Take $D=\co$, so that $f$ can be an $n$-qubit quantum state, and let $G$ be the
Pauli group.
This yields the Pauli-\twrchanged{Local Invertible Map Decision Diagram (LIMDD)}~\cite{limdd}.
Pauli-LIMDDs succinctly represent stabilizer states ---a classically simulatable class~\cite{gottesman1997stabilizer,aaronson2008improved} that is ubiquitous in error correction~\cite{calderbank1997quantum,fowler2012surface}---
while no \twrchanged{EVDD} does~\cite{limdd}.
This expressivity of LIMDD makes it a powerful tool for quantum program analysis, simulation, verification, and synthesis~\cite{vinkhuijzen2023efficient,limdd,hong2026quantum}.

In related work, tensor networks integrate the stabilizer formalism in the same
spirit~\cite{masot2024stabilizer,nakhl2024stabilizer,mello2024clifford,lami2024quantum}.
Tensor trains also admit a
canonical form~\cite{perez2007matrix,vidal2003efficient,schollwock2011density}, but
computing it is expensive, thus only done on demand, %
whereas a decision diagram is canonical by construction~\cite{quist2026tensor}.

\paragraph*{The gap.}
Pauli-LIMDD pays a cubic price to regain canonicity by computing stabilizer groups~\cite{limdd}.
Most work since has gone into reducing that cost:
three independent implementations~\cite{hong2025limtdd,vinkhuijzen2023efficient,sanders2026faster} propose (complete and incomplete) heuristics, and a knowledge compilation map~\cite{vinkhuijzen2024a} settled the tractability of query and manipulation operations. The effort has thus concentrated on making the LIM machinery faster while holding the group
fixed~\cite{sanders2026faster}.

The opposite route is to make the overhead pay off
better by coarsening through a \emph{larger} group $G$ that induces more sharing.
Whether this route leads anywhere is open. Rotation LIMDDs, the case
$k=0$, were explored in~\cite{hong2025limtdd,hong2025advancingLIMDD} with good
practical results but without asymptotic separations. So it is not known whether
enlarging the group buys anything provable, what it costs in canonicity, or whether
it forfeits the tractability of operations. %

We take this second route, which raises three questions about larger groups $G$:
(i) do $G$-LIMDDs still have unique normal forms computable in polynomial time?
(ii) how much succinctness does $G$ buy?
(iii) what queries and operations remain tractable on $G$-LIMDDs?

\paragraph*{Our approach.}
 
We answer the three questions for a
two-parameter family of groups that contains the Pauli case and moves away from it
in two independent directions.
We focus on this family of groups because they allow for unique normal forms computable in polynomial time, 
a property that is essential for the efficiency of LIMDDs.
These groups are (anti-)diagonal and that explains the choice for these groups:
for groups beyond (anti-)diagonal groups we provide a counter example that shows that the normal form 
computation is not local anymore and might require to rebuild the whole diagram.

Let
$R_q=\left(\begin{smallmatrix}1&0\\0&e^{2\pi i/q}\end{smallmatrix}\right)$ be the
phase rotation of order $q$, and let $C^hR_q$ be $R_q$ guarded by $h$ control
qubits, so $C^0R_q=R_q$ and $C^1R_2=CZ$. The set $C^{\leq k}R_q$ collects $C^hR_q$
for all $h\leq k$, and $\gen{C^{\leq k}R_q}$ is generated by placing these matrices
on all choices of qubits. We call $k$ the \emph{control arity} and $q$ the \emph{root order}. The second group within the family adds the bit flip $X$. Up to scalars
$\gen{R_2,X}$ is the Pauli group, so $k=0$, $q=2$ with $X$ is the original Pauli-LIMDD~\cite{limdd},
and $k=0$ with general $q$ gives the rotation diagrams of~\cite{hong2025limtdd}.
Raising $k$ widens the maps from $1$-local to $(k+1)$-local, i.e., gates that act on $k+1$ qubits. Raising $q$ leaves
the class of (controlled) $Z$ gates and generalizes it to general $Z$-rotations.

The two dials are chosen so that the state classes one actually wants land inside the
family. Raising $q$ reaches the diagonal levels of the Clifford
hierarchy~\cite{gottesman1999demonstrating,cui2017diagonal}; raising $k$ reaches the
multi-qubit phase gates of commuting (IQP) circuits~\cite{shepherd2009temporally,bremner2011classical}
and of hypergraph states~\cite{rossi2013quantum}; and adding $X$ retains the stabilizer
states~\cite{gottesman1997stabilizer} that the Pauli case already covered.
\Cref{sec:results} makes each of these precise.
In fact, 
hypergraph states have linear-size diagrams in the family,
and the correspondence is exact: the linear-sized diagrams of the family are, up to
normalization, precisely the hypergraph states.

Moreover, raising the control arity, or multiplying the root order, provides an exponential gain in succinctness, and adding
the bit flip brings a third exponential separation. The quantum Fourier transform matrix is a linear-size tower over plain rotations, without controls.
None of these generalizations costs tractability: the frontier of tractable queries and
transformations is invariant across the whole family. Moreover, for every member of the family, a canonical form can be calculated in polynomial time in the size of the diagram.

In summary, we determine the succinctness order of the family completely
(\Cref{fig:succinctness_picture}). We also give a five-rule reduction system with
unique normal forms for every member of the family, and we settle when the
normal form is computable in polynomial time: for every member of the family,
with bit flips and without, and for every $q$. We settle five
queries and eight transformations that are tractable across the whole family. We show that beyond the (anti-)diagonal
label groups, calculating a normal form that is minimal stops being a local computation. \Cref{sec:prelim}
fixes the definitions, and \Cref{sec:results} states every result of the paper
in full. Full proofs are in later sections.

\subsection{Notations and definitions}\label{sec:prelim}

We write $[n]$ for $\{1,\dots,n\}$ and $[a,b]$ for $\{a,a+1,\dots,b\}$. For a set
$S$ we write $\binom{S}{k}$ for the set of its $k$-element subsets. We put
$\omega_q=e^{2\pi i/q}$.

An $n$-qubit quantum state is a nonzero vector in $\co^{2^n}$. We index its
entries by bit strings, so a state is a function $\{0,1\}^n\to\co$ that assigns an
\emph{amplitude} to every input. We write $\ket x$ for the state that is $1$ on
$x$ and $0$ elsewhere, and $\ket\phi=\sum_x\phi(x)\ket x$ for a general state.
Splitting off the first bit gives the two \emph{cofactors} $\phi_0$ and $\phi_1$
of $\phi$, so that $\ket\phi=\ket0\otimes\ket{\phi_0}+\ket1\otimes\ket{\phi_1}$.

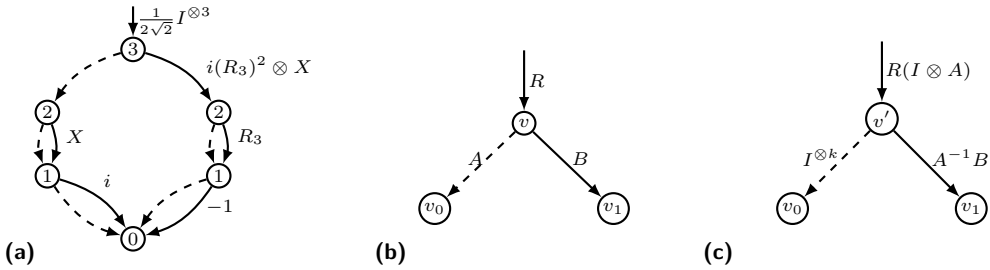
\begin{figure}[b!]
    \centering
    \begin{subfigure}[b]{0.32\textwidth}
    \centering
    \begin{tikzpicture}[scale=0.3,every path/.style={>=latex},inner sep=1pt,minimum size=0.3cm,line width=1pt,node distance=.8cm,thick,font=\scriptsize]
        \node[] (a0) {};
        \node[draw,circle,below = .4cm of a0] (a1) {$3$};
        \node[below = .5cm of a1 ] (a2) {};
        \node[draw,circle,left = of a2] (a2l) {$2$};
        \node[draw,circle,right = of a2] (a2r) {$2$};
        \node[draw,circle,below = .5cm of a2l] (a3l) {$1$};
        \node[draw,circle,below = .5cm of a2r] (a3r) {$1$};
        \node[below = .5cm  of a3l] (a4l) {};
        \node[draw,circle,right = of a4l] (a4) {$0$};
        \draw[e1] (a0) edge node[right] {$\frac{1}{2\sqrt{2}}I^{\otimes3}$} (a1);
        \draw[e0=20] (a1) edge (a2l);
        \draw[e1=20] (a1) edge node[above right,pos=.7] {~$i(R_3)^2\otimes X$} (a2r);
        \draw[e0=20] (a2l) edge (a3l);
        \draw[e1=20] (a2l) edge node[right,pos=.3] {~$X$} (a3l);
        \draw[e0=20] (a2r) edge (a3r);
        \draw[e1=20] (a2r) edge node[right,pos=.3] {~$R_3$} (a3r);
        \draw[e0=20] (a3l) edge (a4);
        \draw[e1=20] (a3l) edge node[above right,pos=.4] {~$i$} (a4);
        \draw[e0=20] (a3r) edge (a4);
        \draw[e1=20] (a3r) edge node[right,pos=.3] {~$-1$} (a4);
    \end{tikzpicture}\vspace{-1em}
    \caption{}\label{fig:limdd-a}
    \end{subfigure}
    \hfill
    \begin{subfigure}[b]{0.28\textwidth}
    \centering
    \begin{tikzpicture}[scale=0.3,every path/.style={>=latex},inner sep=1pt,minimum size=0.3cm,line width=1pt,node distance=.8cm,thick,font=\scriptsize]
        \node[] (a0) {};
        \node[draw,circle,below = of a0] (a1) {$v$};
        \node[below = of a1 ] (a2) {};
        \node[draw,circle,left = of a2] (a2l) {$v_0$};
        \node[draw,circle,right = of a2] (a2r) {$v_1$};
        \draw[e1] (a0) edge node[right] {$R$} (a1);
        \draw[e0] (a1) edge node[left,pos=.4] {~$A$} (a2l);
        \draw[e1] (a1) edge node[right,pos=.4] {~$B$} (a2r);
    \end{tikzpicture}
    \caption{}\label{fig:limdd-b}
    \end{subfigure}
    \hfill
    \begin{subfigure}[b]{0.34\textwidth}
    \centering
    \begin{tikzpicture}[scale=0.3,every path/.style={>=latex},inner sep=1pt,minimum size=0.3cm,line width=1pt,node distance=.8cm,thick,font=\scriptsize]
        \node[] (a0) {};
        \node[draw,circle,below = of a0] (a1) {$v'$};
        \node[below = of a1 ] (a2) {};
        \node[draw,circle,left = of a2] (a2l) {$v_0$};
        \node[draw,circle,right = of a2] (a2r) {$v_1$};
        \draw[e1] (a0) edge node[right] {$R(I\otimes A)$} (a1);
        \draw[e0] (a1) edge node[left,pos=.4] {~$I^{\otimes k}$} (a2l);
        \draw[e1] (a1) edge node[right,pos=.4] {~$A^{-1}B$} (a2r);
    \end{tikzpicture}
    \caption{}\label{fig:limdd-c}
    \end{subfigure}
    \caption{\textbf{\subref{fig:limdd-a})}: a $\gen{R_3,X}$-LIMDD of
    $\frac{1}{2\sqrt2}\bigl(\ket{000}+i\ket{001}+i\ket{010}+\ket{011}-i\ket{100}+i\ket{101}-i\ket{110}+i\omega_3^2\ket{111}\bigr)$.
    Nodes carry their index; $0$-edges are dotted and unlabelled edges carry the
    identity.\\ \textbf{\subref{fig:limdd-b}), \subref{fig:limdd-c})}: the
    $0$-edge normal form. The two segments represent the same state, so every
    $0$-edge label may be assumed to be the identity. Here $v,v'$ have index $k+1$
    and $v_0,v_1$ index~$k$.}
    \label{fig:limdd}
\end{figure}

\paragraph*{Decision diagrams and their edge labels.}

We follow the presentation of quantum decision diagrams
in~\cite{vinkhuijzen2024a}, in which the set that edge labels are drawn from is a
\emph{parameter} of the data structure. Fix a set $\mathcal L$ of invertible
matrices. An \emph{$\mathcal L$-labeled decision diagram} is a finite rooted
directed acyclic graph with one root and one leaf. Every node $v$ carries an index
$id(v)\in[0,n]$, the root has index $n$, and the leaf has index $0$. The root has
one incoming edge with no source, the \emph{root edge}. Every non-leaf node has two
outgoing edges, the \emph{$0$-edge} and the \emph{$1$-edge}. Every edge $e$ carries
a label $label(e)\in\mathcal L$ or the constant $0$. No node is skipped on a path,
so $(v,w)\in E$ implies $id(v)=id(w)+1$.
\twrchanged{
  (Technically, $\mathcal L$-labeled decision diagrams generalize quasi-OBDDs
  \cite[p.\ 50-51]{Book:Wegener00}.)
}
Let $G$ be a group of invertible $2^n\times2^n$ matrices. A \emph{$G$-LIM} is a matrix
$\lambda\cdot g$ with $\lambda\in\co\setminus\{0\}$ and $g\in\gen G_n$, and a
\emph{$\gen G$-LIMDD} is a diagram whose labels are $G$-LIMs. The scalar must be
nonzero, since a LIM is invertible by definition. The label $0$, used by rules R1
and R3 of \Cref{sec:canonicity}, is therefore not a LIM but a separate case, and no
inverse of it is ever taken. Taking the trivial group recovers the
\emph{edge-valued} decision diagram (EVDD), whose labels are the nonzero scalars
alone; taking the Pauli group recovers the LIMDD of~\cite{limdd}, which makes every
stabilizer state linear in size, something no EVDD achieves~\cite{limdd}. The label
group is thus a parameter of the data structure, and it is the parameter we vary:
\Cref{sec:results} takes $\gen G$ from a two-parameter family containing the Pauli
case. Algebraic and multi-terminal
\twrchanged{
OBDDs~\cite{bahar1997algebric,clarke1993spectral}
}
go the other way and keep the labels trivial, putting the amplitudes on several
leaves instead; one leaf suffices here because every edge already carries a scalar.

The \emph{size} of a diagram is its number of nodes, and its \emph{width} at level
$j$ is the number of nodes of index $j$. A \emph{Tower}-LIMDD is one in which both
edges of every node point to the same child, so it has exactly $n+1$ nodes. Two
nodes $u,v$ are \emph{$\gen G$-equivalent} when $\ket u=\ell\ket v$ for some
$G$-LIM $\ell$, and the \emph{stabilizer}
$Stab(v):=\{g\in\gen G\mid g\ket v=\ket v\}$ is a subgroup of $\gen G$.
\Cref{fig:limdd} shows an example, together with the $0$-edge normal form
established below.

\paragraph*{Gates and groups.}

It remains to say which groups $\gen G$ we put in the label slot. We use the Pauli
matrices $I$, $X=\left(\begin{smallmatrix}0&1\\1&0\end{smallmatrix}\right)$,
$Y=\left(\begin{smallmatrix}0&-i\\i&0\end{smallmatrix}\right)$,
$Z=\left(\begin{smallmatrix}1&0\\0&-1\end{smallmatrix}\right)$, the Hadamard matrix
$H=\tfrac1{\sqrt2}\left(\begin{smallmatrix}1&1\\1&-1\end{smallmatrix}\right)$, and
the phase rotations
$R(\theta)=\left(\begin{smallmatrix}1&0\\0&e^{i\theta}\end{smallmatrix}\right)$ and
$R_q:=R(2\pi/q)=\left(\begin{smallmatrix}1&0\\0&\omega_q\end{smallmatrix}\right)$
for $q\in\natur$, so that $R_1=I$ and $R_2=Z$. For a one-qubit matrix $U$ and
$h\geq0$, the matrix
$C^hU=\operatorname{diag}(1,\dots,1,U)$ of size $2^{h+1}\times2^{h+1}$ applies $U$
to its last qubit exactly when its first $h$ qubits are all $1$; thus $C^0U=U$ and
$C^1Z=CZ$. In particular
$C^hR_q=\operatorname{diag}(1,\dots,1,\omega_q)$, and on basis states
$C^hR_q\ket{x_1\dots x_{h+1}}=\omega_q^{\prod_{j=1}^{h+1}x_j}\ket{x_1\dots x_{h+1}}$.

We place such a matrix on chosen qubits of a larger register. Let $\seq(n,m)$ be
the set of non-repeating sequences of length $m$ over $[n]$. For a
$2^m\times2^m$ matrix $g$ and $S=(a_1,\dots,a_m)\in\seq(n,m)$ we write
$g_{n,S}=P_S^\dagger(g\otimes I^{\otimes(n-m)})P_S$, where
$P_S\ket{x_1\dots x_n}=\ket{x_{\pi(1)}\dots x_{\pi(n)}}$ for any permutation $\pi$
of $[n]$ with $\pi(a_j)=j$; the action of $\pi$ outside $S$ is irrelevant. We drop
the subscript $n$ when it is clear from the context. For example
$(C^2U)_{5,(2,3,5)}$ is the five-qubit matrix in which qubits $2$ and $3$ control
$U$ on qubit $5$. Since $C^hR_q$ is invariant under permuting the qubits the gate is applied to, we
write $(C^{\ell-1}R_q)_{(a_1,\dots,a_\ell)}$ for a \emph{set}
$a=\{a_1<\dots<a_\ell\}$ and call $\ell=|a|$ its \emph{arity}.

For a $2^m\times2^m$ matrix $U$ and $n\geq m$, let $\gen U_n$ be the group
generated by all $U_{n,S}$ with $S\in\seq(n,m)$, and likewise
$\gen{U_1,\dots,U_\ell}_n$ for several generators. We drop $n$ when it is clear
from the context. This paper studies the two families
\begin{equation*}
  \gen{C^{\leq k}R_q}:=\gen{R_q,CR_q,\dots,C^kR_q},
  \qquad
  \gen{C^{\leq k}R_q,X}:=\gen{R_q,CR_q,\dots,C^kR_q,X},
\end{equation*}
for $k\in\z_{\geq0}$ and $q\in\z_{\geq1}$, with $\gen{C^{\leq0}R_q}=\gen{R_q}$. We
call $k$ the \emph{control arity} and $q$ the \emph{root order}. Up to scalars,
$\gen{C^{\leq0}R_2,X}=\gen{Z,X}$ is the Pauli group, so the original
LIMDD~\cite{limdd} is the case $k=0$, $q=2$ with $X$.

The bit flip does not commute with the rotations, but it moves past them at a
bounded cost. This identity underlies every statement about $\gen{C^{\leq k}R_q,X}$
in this paper, so we record it first.

\begin{lemma}[commutation]\label{lem:commute}
  Let $a=\{a_1<\dots<a_\ell\}\subseteq[n]$ and $j\in[\ell]$. Then
  \begin{equation}\label{eq:commute}
    X_{(a_j)}(C^{\ell-1}R_q)_{(a)}
    =\left[(C^{\ell-1}R_q)_{(a)}\right]^{q-1}
     (C^{\ell-2}R_q)_{(a\setminus\{a_j\})}\,X_{(a_j)},
  \end{equation}
  where $C^{-1}R_q:=\omega_q$. In particular, for $\ell=1$ and any $d$,
  \begin{equation}\label{eq:R_q--X__commutation}
    X_{(a)}\left[(R_q)_{(a)}\right]^{d}=\omega_q^{d}\left[(R_q)_{(a)}\right]^{-d}X_{(a)} .
  \end{equation}
  Consequently the group $\gen{C^{\leq k}R_q}$ is normal in $\gen{C^{\leq k}R_q,X}$; i.e., $ghg^{-1}\in\gen{C^{\leq k}R_q}$ for all $g\in\gen{C^{\leq k}R_q,X}$ and $h\in\gen{C^{\leq k}R_q}$.
\end{lemma}

\paragraph*{The \texorpdfstring{$0$}{0}-edge normal form.}

Every $0$-edge label may be assumed to be $I^{\otimes n}$. If a node $v$ has
$0$-edge label $A\neq0$, then
$A\ket{v_0}\otimes\ket0+B\ket{v_1}\otimes\ket1=(I\otimes A)(\ket0\otimes\ket{v_0}+\ket1\otimes A^{-1}B\ket{v_1})$,
so replacing $B$ by $A^{-1}B$, setting the $0$-edge label to the identity, and
multiplying every incoming edge label by $I\otimes A$ leaves the represented state
unchanged. In case $X\not\in\gen{G}$ and $A=0$, we may assume the 1-edge label to be $B=I^{\otimes n}$ by a similar reason.

Finally we record why we close the generators under \emph{all} arities up to $k$,
rather than fixing the arity at $k$. Consider $\gen{CZ}$, which is not
arity-closed. The states $\ket\phi=\ket0+\ket1$ and $\ket\psi=\ket0-\ket1$ are not
$\gen{CZ}$-equivalent, because $\gen{CZ}$ acts trivially on one qubit. Prepending a
qubit makes them equivalent, because
$\ket1\otimes\ket\phi=CZ(\ket1\otimes\ket\psi)$. So two nodes that cannot be
merged may become mergeable once a node is placed above them. A diagram can then
shrink when a qubit is added, which defeats both a local $\mathsf{makeEdge}$ and
any lower-bound argument by layers. Arity-closed groups do not have this defect.\footnote{Interestingly, in the canonical form, all edge labels have full arity. I.e., every edge label is in $\gen{C^kR_q}$ or, if $X\in\gen{G}$, in $\gen{C^kR_q,X}$, as shown in \Cref{sec:canonicity}.}

\begin{lemma}[arity closure]\label{lem:arity}
  Let $G=\gen{C^{\leq k}R_q}$ or $G=\gen{C^{\leq k}R_q,X}$, and let $u,w$ be nodes
  of index $n-1$ that are not $\gen G$-equivalent. Then no two nodes of index $n$
  whose cofactors are $u$ and $w$, respectively, become $\gen G$-equivalent.
\end{lemma}

The proof is a direct application of the layer split of \Cref{sec:machinery}
and is proved in \Cref{app:prelim}.

\subsection{Results}\label{sec:results}

\subsubsection{The succinctness order.}
A class $L$ of representations of quantum states is \emph{fully expressive} when
every state has a representation in $L$. For fully expressive $L_1,L_2$ we write
$L_2\leq_sL_1$, and say $L_2$ is \emph{at least as succinct as} $L_1$, when there
is a polynomial $p$ such that every $C_1\in L_1$ has an equivalent $C_2\in L_2$
with $|C_2|\leq p(|C_1|)$. We write $L_2<_sL_1$ when $L_2\leq_sL_1$ but
$L_1\not\leq_sL_2$, and we call $L_1$ and $L_2$ \emph{incomparable} when both $L_2\nleq_sL_1$ and 
$L_1\nleq_sL_2$. We say $L_2$ is \emph{exponentially more succinct} than $L_1$ when
$L_2<_sL_1$ and there is a family of states whose representations in $L_2$ are
exponentially smaller than any of their representations in $L_1$.

The succinctness order of the family (\Cref{sec:succinctness}) measures what the
coarsening buys. For all $h<k$, there are states with a linear-size
Tower-$\gen{C^{\leq k}R_q}$-diagram whose $\gen{C^{\leq h}R_q}$-diagram has size
$2^{\Omega(n)}$ under every variable order. Moreover, for all $p$ that do not divide $q$, there are states with a linear-size
Tower-$\gen{C^{\leq k}R_p}$-diagram whose $\gen{C^{\leq k}R_q}$-diagram has size
$2^{\Omega(n)}$ under every variable order. So control arity and root order each
induce a strict exponential hierarchy, and orders neither of which divides the other are
incomparable. The separating objects are hypergraph states.

\begin{figure}[tb]
\centering
\colorlet{cq}{green!45!black}
\colorlet{ck}{blue!65!black}
\colorlet{cx}{red!78!black}
\begin{tikzpicture}[
  qn/.style={font=\tiny,fill=white,inner sep=1pt},
  qedge/.style={-latex,semithick,cq},
  kedge/.style={-latex,semithick,ck},
  xedge/.style={-latex,semithick,cx},
  weak/.style={densely dashed},
  gbox/.style={draw=black!40,dashed,rounded corners=2pt},
]

\def\boxlist{%
  b0/0/0/{}/{}/black!4/0/1.7,%
  b1/0.4/3.0/{C^{\leq1}}/{}/black!4/0/1.7,%
  b2/0.8/6.0/{C^{\leq2}}/{}/black!4/0/1.7,%
  x0/5.7/1.3/{}/{,X}/cx!4/1/2.15,%
  x1/6.1/4.3/{C^{\leq1}}/{,X}/cx!4/1/2.15,%
  x2/6.5/7.3/{C^{\leq2}}/{,X}/cx!4/1/2.15%
}

\foreach \bn/\ox/\oy/\pre/\suf/\fc/\xf/\s in \boxlist {
  \path[gbox,fill=\fc] (\ox-0.8,\oy-0.3) rectangle (\ox+2*\s+0.8,\oy+1.7);
  \node[qn] (\bn-R2)  at (\ox,\oy)             {$\gen{\pre R_2\suf}$};
  \node[qn] (\bn-R3)  at (\ox+\s,\oy)          {$\gen{\pre R_3\suf}$};
  \node[qn] (\bn-R5)  at (\ox+2*\s,\oy)        {$\gen{\pre R_5\suf}$};
  \node[qn] (\bn-R4)  at (\ox+0.5*\s,\oy+0.8) {$\gen{\pre R_4\suf}$};
  \node[qn] (\bn-R6)  at (\ox+1.5*\s,\oy+0.8) {$\gen{\pre R_6\suf}$};
  \node[qn] (\bn-Rot) at (\ox+2*\s,\oy+2.3)    {$\gen{\pre\Rot\suf}$};
}
\node[qn] (evdd)  at (1.7,-1.0)  {EVDD};
\node[qn] (xonly) at (7.85,0.3) {$\gen X$};

\foreach \bn/\ox/\oy/\pre/\suf/\fc/\xf/\s in \boxlist {
  \draw[qedge] (\bn-R2) -- (\bn-R4);
  \draw[qedge] (\bn-R2) -- (\bn-R6);
  \draw[qedge] (\bn-R3) -- (\bn-R6);
  \draw[qedge] (\bn-R4) -- (\bn-Rot);
  \draw[qedge] (\bn-R6) -- (\bn-Rot);
  \draw[qedge] (\bn-R5) -- (\bn-Rot);
}
\foreach \n in {R2,R3,R4,R5,R6,Rot}{
  \draw[kedge]            (b0-\n) -- (b1-\n);
  \draw[kedge]            (b1-\n) -- (b2-\n);
  \draw[kedge] (x0-\n) -- (x1-\n);
  \draw[kedge] (x1-\n) -- (x2-\n);
  \draw[xedge]            (b0-\n) -- (x0-\n);
  \draw[xedge]            (b1-\n) -- (x1-\n);
  \draw[xedge]            (b2-\n) -- (x2-\n);
}
\draw[qedge] (evdd) -- (b0-R2);
\draw[qedge] (evdd) -- (b0-R3);
\draw[qedge] (evdd) -- (b0-R5);
\draw[qedge] (xonly) -- (x0-R2);
\draw[qedge] (xonly) -- (x0-R3);
\draw[qedge] (xonly) -- (x0-R5);
\draw[xedge] (evdd) -- (xonly);

\foreach \bn/\ox/\oy/\pre/\suf/\fc/\xf/\s in \boxlist {
  \node[qn] at (\bn-R2.center)  {$\gen{\pre R_2\suf}$};
  \node[qn] at (\bn-R3.center)  {$\gen{\pre R_3\suf}$};
  \node[qn] at (\bn-R5.center)  {$\gen{\pre R_5\suf}$};
  \node[qn] at (\bn-R4.center)  {$\gen{\pre R_4\suf}$};
  \node[qn] at (\bn-R6.center)  {$\gen{\pre R_6\suf}$};
  \node[qn] at (\bn-Rot.center) {$\gen{\pre\Rot\suf}$};
}
\node[qn] at (evdd.center)  {EVDD};
\node[qn] at (xonly.center) {$\gen X$};
\end{tikzpicture}
\caption{Succinctness map of the LIMDD family in its three directions: root
order $q$ (green), control arity $k$ (blue), and the bit
flip $X$ to the right (red); the dashed boxes collect the finite root orders,
with $\Rot$ drawn above each of them. An arrow $L_1\rightarrow L_2$ means that $L_2$ is
exponentially more succinct than $L_1$; every arrow drawn is a proved exponential
separation, and classes joined by no path are incomparable.
In the diagonal column (left), green arrows are
\Cref{thm:hypergraphstate_is_tower,theorem:prime-hierarchy,cor:continuous} and blue
arrows \Cref{thm:hypergraphstate_is_tower,theorem:prime-hierarchy}; red arrows are
\Cref{cor:X}. In the bit-flip column (right), the green and blue arrows and the
incomparabilities are \Cref{cor:Xhierarchy}, proved by biasing the amplitudes of the same
hypergraph family (\Cref{thm:Xcolumn}); the arrows into the $\Rot$ row of either column use
the coefficient form of the argument (\Cref{rem:Xhonest}).}
\label{fig:succinctness_picture}
\end{figure}
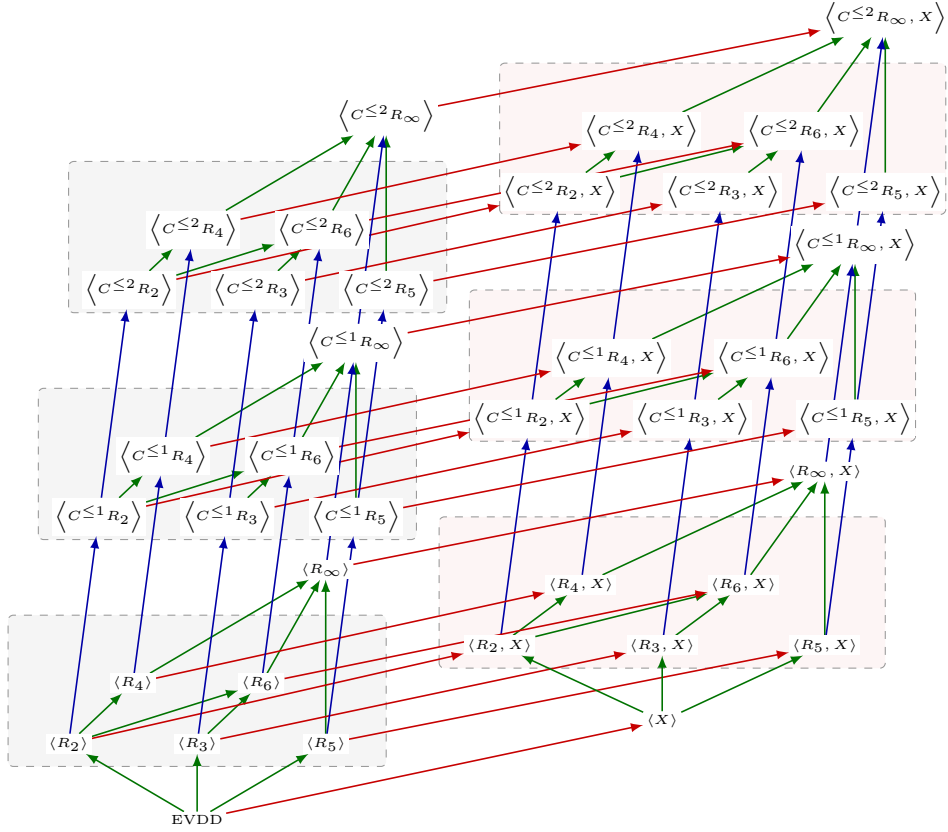

A \emph{hypergraph} is a pair $G=(V,E)$ in which every $e\in E$ is a subset of $V$.
Multi-edges are allowed, so $E$ is a multiset. A \emph{$k$-hypergraph} is one in
which every edge has size at most $k$. 

\begin{definition}\label{def:hypergraphstate}
  The \emph{hypergraph $p$-state} of a hypergraph $G=(V,E)$ with $V=\{x_1,\dots,x_m\}$
  is the $m$-qubit state
  \begin{equation*}
    \ket{\phi^G_p}
      =\frac{1}{\sqrt{2^m}}\sum_{x\in\{0,1\}^m}\ \prod_{e\in E}\omega_p^{\prod_{x_i\in e}x_i}\ \ket x
      =\frac{1}{\sqrt{2^m}}\prod_{e\in E}(C^{|e|-1}R_p)_{(e)}\ \bigl(\ket0+\ket1\bigr)^{\otimes m} .
  \end{equation*}
\end{definition}
This definition is 
more general than the usual one~\cite{qu2012encoding,qu2013multipartite,qu2013relationship,rossi2013quantum},
which fixes $p=2$. As an example, 
$\frac{1}{\sqrt{2^5}}(C^1R_p)_{(1,2)}(C^1R_p)^2_{(3,5)}(C^2R_p)^2_{(1,2,4)}(C^4R_p)_{(2,3,4,5)}\ket+^{\otimes 5}$
is the $p$-state for the hypergraph shown Figure~\ref{fig:hypergraph}.

\begin{restatable}{theorem}{thmhgtower}\label{thm:hypergraphstate_is_tower}
  Let $k,p\in\natur$ and let $G$ be a $k$-hypergraph. Then $\ket{\phi^G_p}$ is
  represented by a Tower-$\gen{C^{\leq k-2}R_p}$-LIMDD, for every variable order.
\end{restatable}

The proof, an induction that splits off one vertex per level, is in
\Cref{sec:towers}.

The correspondence runs in both directions, which is what makes the family a
characterization rather than just a witness.

\begin{theorem}\label{thm:tower_is_hypergraphstate}
  Every Tower-$\gen{C^{\leq k-2}R_p}$-LIMDD in which every edge label has scalar
  component $1$ represents an (unnormalized)
  $k$-hypergraph $p$-state. If the root edge scalar is set to the normalization constant $2^{-n/2}$, the $k$-hypergraph $p$-state is normalized.
\end{theorem}

This result establishes the succinct representation in a data structure that allows for efficient application of operations of two familiar classes of quantum states at
once.
First, the diagrams with a linear number of nodes are exactly the states
produced by a commuting (IQP) circuit before its final Hadamard layer. Such a circuit
is interesting because, while it can be simulated in a class believed to be strictly
weaker than BQP, its output distribution is nonetheless \#P-hard to compute
exactly~\cite{bremner2017achieving}. Formally, an IQP circuit has the form
$H^{\otimes n}DH^{\otimes n}$ with $D$ diagonal~\cite{bremner2017achieving}, and
$D\ket{+}^{\otimes n}$ with $D$ of degree $\leq k+2$ and angles in $(2\pi/q)\z$ is
exactly the hypergraph $q$-state of \Cref{def:hypergraphstate}. Removing the final
Hadamard layer therefore identifies the pre-measurement states of IQP circuits with
the hypergraph states of this diagram family.
Second, the family absorbs an entire diagonal layer at no cost in size: every gate in
such a layer is captured by a single LIM on the root edge of the rotation LIMDD, without creating any new
node. On qubits, a diagonal gate at level $w$ of the Clifford hierarchy is a product
of controlled phase rotations of order dividing $2^w$~\cite{cui2017diagonal}, so for
$2^w \mid q$ it is a single $\gen{C^{\leq w-1}R_q}$-LIM. The levels of this
construction thus form a hierarchy of states --- which we call
\emph{diagonal-hierarchy states} --- generated by these controlled-rotation LIMs. 
This hierarchy is weaker than the full Clifford hierarchy but has been
studied extensively. Pauli-LIMDDs cannot represent these diagonal-hierarchy states
with polynomial size, whereas rotation-LIMDDs represent them efficiently. 
Concretely, since a
stabilizer state is a Pauli tower~\cite{limdd}, which for even $q$ is already a
diagram of this family, and since a LIM on the root edge of the rotation LIMDD creates no node, every state
$D\ket S$ with $\ket S$ stabilizer and $D$ diagonal of level $w$ is an $(n+1)$-node
tower, independently of the $T$-count of $D$. Instances of this construction include
a transversal $T$ or $CCZ$ gate applied to any stabilizer state, as prepared in
magic-state distillation~\cite{bravyi2012magic,paetznick2013universal} and
injection~\cite{bravyi2005universal}.

To prove the $\gen{C^{\leq h}R_q}$ hierarchy 
we use \emph{$k$-hyperized graphs}. The \emph{$k$-hyperization} of a graph
$G'=(V',E')$ is the $k$-hypergraph $G=(V,E)$ obtained by replacing every $v\in V'$
by $k$ vertices $v_1,\dots,v_k$ and extending every edge over all replacements:
\begin{equation*}
  V=\bigcup_{v\in V'}\{v_1,\dots,v_k\},\qquad
  E=\bigcup_{\{v,w\}\in E'}\ \bigcup_{i=1}^{k}
     \bigl\{\ \{v_1,\dots,v_k,w_i\},\ \{v_i,w_1,\dots,w_k\}\ \bigr\}.
\end{equation*}
For instance, for a $k = 3$, each edge $\{v,w\}$ in the graph contributes the following six hyperedges.
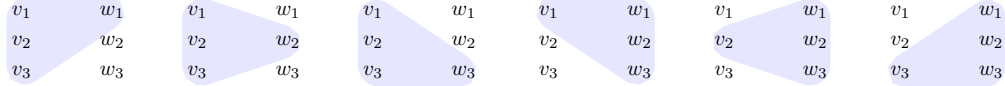
\begin{figure}[h]
\begin{subfigure}{0.166\textwidth}
\centering
\begin{tikzpicture}[scale=0.8, every node/.style={scale=0.8}]		
\fill[color=blue!10,line width=4mm,draw=blue!10,rounded
			corners=3pt] (0,-0.5) -- (0,+0.5) -- (1.5,+0.5) -- cycle;
			
\node (v1) at (0,+0.5) {$v_1$};
\node (v2) at (0,0) {$v_2$};
\node (v3) at (0,-0.5) {$v_3$};
\node (w1) at (1.5,+0.5) {$w_1$};
\node (w2) at (1.5,0) {$w_2$};
\node (w3) at (1.5,-0.5) {$w_3$};
\end{tikzpicture}
\end{subfigure}\begin{subfigure}{0.166\textwidth}
\centering
\begin{tikzpicture}[scale=0.8, every node/.style={scale=0.8}]		
\fill[color=blue!10,line width=4mm,draw=blue!10,rounded
			corners=3pt] (0,-0.5) -- (0,+0.5) -- (1.5,0) -- cycle;
			
\node (v1) at (0,+0.5) {$v_1$};
\node (v2) at (0,0) {$v_2$};
\node (v3) at (0,-0.5) {$v_3$};
\node (w1) at (1.5,+0.5) {$w_1$};
\node (w2) at (1.5,0) {$w_2$};
\node (w3) at (1.5,-0.5) {$w_3$};
\end{tikzpicture}
\end{subfigure}\begin{subfigure}{0.166\textwidth}
\centering
\begin{tikzpicture}[scale=0.8, every node/.style={scale=0.8}]				
\fill[color=blue!10,line width=4mm,draw=blue!10,rounded
			corners=3pt] (0,-0.5) -- (0,+0.5) -- (1.5,-0.5) -- cycle;
			
\node (v1) at (0,+0.5) {$v_1$};
\node (v2) at (0,0) {$v_2$};
\node (v3) at (0,-0.5) {$v_3$};
\node (w1) at (1.5,+0.5) {$w_1$};
\node (w2) at (1.5,0) {$w_2$};
\node (w3) at (1.5,-0.5) {$w_3$};
\end{tikzpicture}
\end{subfigure}\begin{subfigure}{0.166\textwidth}
\centering
\begin{tikzpicture}[scale=0.8, every node/.style={scale=0.8}]			
\fill[color=blue!10,line width=4mm,draw=blue!10,rounded
			corners=3pt] (1.5,-0.5) -- (1.5,+0.5) -- (0,+0.5) -- cycle;
			
\node (v1) at (0,+0.5) {$v_1$};
\node (v2) at (0,0) {$v_2$};
\node (v3) at (0,-0.5) {$v_3$};
\node (w1) at (1.5,+0.5) {$w_1$};
\node (w2) at (1.5,0) {$w_2$};
\node (w3) at (1.5,-0.5) {$w_3$};
\end{tikzpicture}
\end{subfigure}\begin{subfigure}{0.166\textwidth}
\centering
\begin{tikzpicture}[scale=0.8, every node/.style={scale=0.8}]			
\fill[color=blue!10,line width=4mm,draw=blue!10,rounded
			corners=3pt] (1.5,-0.5) -- (1.5,+0.5) -- (0,0) -- cycle;
			
\node (v1) at (0,+0.5) {$v_1$};
\node (v2) at (0,0) {$v_2$};
\node (v3) at (0,-0.5) {$v_3$};
\node (w1) at (1.5,+0.5) {$w_1$};
\node (w2) at (1.5,0) {$w_2$};
\node (w3) at (1.5,-0.5) {$w_3$};
\end{tikzpicture}
\end{subfigure}\begin{subfigure}{0.166\textwidth}
\centering
\begin{tikzpicture}[scale=0.8, every node/.style={scale=0.8}]			
\fill[color=blue!10,line width=4mm,draw=blue!10,rounded
			corners=3pt] (1.5,-0.5) -- (1.5,+0.5) -- (0,-0.5) -- cycle;
			
\node (v1) at (0,+0.5) {$v_1$};
\node (v2) at (0,0) {$v_2$};
\node (v3) at (0,-0.5) {$v_3$};
\node (w1) at (1.5,+0.5) {$w_1$};
\node (w2) at (1.5,0) {$w_2$};
\node (w3) at (1.5,-0.5) {$w_3$};
\end{tikzpicture}
\end{subfigure}
\end{figure}

Our separation uses the $k$-hyperization of the $n\times n$ grid graph $G^k_n$. We write
$\ket{\psi^n_{k,p}}$ for the hypergraph state $\ket{\phi^{G^k_n}_p}$; it is
prepared from $\ket+^{\otimes kn^2}$ by $O(kn^2)$ gates $C^{k}R_p$.

\begin{restatable}{theorem}{thmprimehierarchy}\label{theorem:prime-hierarchy}
  Let $n,k,p,h,q\in\natur$. If $h<k$, or if $q$ is not a multiple of $p$, then
  $\ket{\psi^n_{k+1,p}}$ needs a $\gen{C^{\leq h}R_q}$-LIMDD of size
  $2^{\Omega(n)}$, for every variable order.
\end{restatable}

The proof, in \Cref{sec:lowerbound}, is a counting argument based on an induced matching in the grid graph (\Cref{lem:matching}). An easy corollary of Theorem~\ref{theorem:prime-hierarchy} is an exponentional separation between 
$\gen{C^{\leq k}R_q}$-LIMDDs and $C^{\leq k}\Rot$-LIMDDs, where $\Rot=\{\operatorname{diag}(1,e^{2\pi i\alpha})\mid\alpha\in\mathbb R\}$ and 
$C^{\leq k}\Rot$ consist of the matrices $C^hU$ with $h\leq k$ and $U\in\Rot$.

\begin{restatable}{corollary}{corContinuous}\label{cor:continuous}
  For every $k,q\in\natur$, the class of $C^{\leq k}\Rot$-LIMDDs is exponentially
  more succinct than the class of $\gen{C^{\leq k}R_q}$-LIMDDs.
\end{restatable}

The combination of Theorems~\ref{thm:hypergraphstate_is_tower} and~\ref{theorem:prime-hierarchy} yields the succinctness hierarchy of $\gen{C^{\leq k}R_q}$-LIMDDs depicted Figure~\ref{fig:succinctness_picture}. An analogous hierarchy holds for $\gen{C^{\leq k}R_q,X}$-LIMDDs, where the bit flip operation $X$ is allowed. The variant of \Cref{theorem:prime-hierarchy} for $\gen{C^{\leq k}R_q,X}$-LIMDDs leans on a generalization of hypergraph states that we call \emph{$r$-biased} hypergraph states.

\begin{definition}\label{def:biasedhypergraphstate}
  Let $G=(V,E)$ be a hypergraph with $V=\{x_1,\dots,x_m\}$, let $p\in\natur$ and let
  $r\in\co\setminus\{0\}$. The \emph{$r$-biased} hypergraph state of $G$ is
  \begin{align*}
    \ket{\phi^G_{p,r}}&=(1{+}|r|^2)^{-m/2}\prod_{e\in E}(C^{|e|-1}R_p)_{(e)}\ \bigl(\ket0+r\ket1\bigr)^{\otimes m} .
  \end{align*}
  The amplitude at $x$ is $r^{|x|}\prod_{e\in E}\omega_p^{\prod_{x_i\in e}x_i}$, up to
  normalization, with $|x|$ the Hamming weight of $x$. Setting $r=1$ yields
  $\ket{\phi^G_p}$. We write $\ket{\psi^n_{k,p,r}}$ for the $r$-biased hypergraph $p$-state where the hypergraph is the $k$-hyperization of the $n\times n$ grid.
\end{definition}

Since every LIM carries a free scalar, the tower $\gen{C^{\leq k-2}R_p}$-LIMDD of Theorem~\ref{thm:hypergraphstate_is_tower} for $\ket{\phi^G_{p}}$ can be turned into a tower $\gen{C^{\leq k-2}R_p}$-LIMDD for $\ket{\phi^G_{p,r}}$ by changing the scalar of the LIMs on its edges (\Cref{lem:biasedtower}). So biased hypergraph states are easy to represent as $\gen{C^{\leq k}R_p}$-LIMDD, and thus as $\gen{C^{\leq k}R_p,X}$-LIMDD for large enough $k$ and $p$. On the other hand, using $|r|\neq 1$ in $\ket{\psi^n_{k,p,r}}$ disables the bit flip $X$ in a sense. Every cofactor of $\ket{\phi^G_{p,r}}$, in every variable order, has amplitude moduli
$c\,|r|^{|y|}$ --- fixing a prefix only multiplies them all by one $|r|^{|\vec\alpha|}$
--- while a bit flip $X^{(a)}$ on qubits $a$ would change the exponent by $|y|-|y\oplus a|$, which
jumps by $2$ across a flipped coordinate. So $X$ cannot occur, every equivalence between
cofactors is diagonal (\Cref{lem:rigidity}), and we can reuse the argument of Theorem~\ref{theorem:prime-hierarchy}.

\begin{theorem}[restate=thmxcolumn]\label{thm:Xcolumn}
  Let $n,k,h,p,q\in\natur$ and let $r\in\co$ with $r\neq0$ and $|r|\neq1$. If $h<k$
  and $q=p\geq2$, or if $q$ is not a multiple of $p$, then $\ket{\psi^n_{k+1,p,r}}$
  needs a $\gen{C^{\leq h}R_q,X}$-LIMDD of size $2^{\Omega(n)}$, for every variable
  order.
\end{theorem}

Together with \Cref{lem:biasedtower} this gives the succinctness hierarchy for $\gen{C^{\leq k}R_q,X}$-LIMDDs.

\begin{corollary}\label{cor:Xhierarchy}
  For all $k,h,p,q\in\natur$:
  \begin{enumerate}
    \item if $h<k$ and $q\geq2$ then $\gen{C^{\leq k}R_q,X}$-LIMDD is exponentially
      more succinct than $\gen{C^{\leq h}R_q,X}$-LIMDD;
    \item if $p\mid q$ and $p<q$ then $\gen{C^{\leq k}R_q,X}$-LIMDD is exponentially
      more succinct than $\gen{C^{\leq k}R_p,X}$-LIMDD --- with $p=1$,
      $\gen{R_q,X}$-LIMDD is exponentially more succinct than $\gen X$-LIMDD;
    \item $\gen{C^{\leq k}\Rot,X}$-LIMDD is exponentially more succinct than
      $\gen{C^{\leq k}R_q,X}$-LIMDD; and
    \item if $p\nmid q$ and $q\nmid p$ then $\gen{C^{\leq k}R_p,X}$-LIMDD and
      $\gen{C^{\leq h}R_q,X}$-LIMDD are incomparable, for all $k,h$.
  \end{enumerate}
  Together with \Cref{cor:X} this determines the succinctness order of the whole family, so
  \Cref{fig:succinctness_picture} has no unproved edge left.
\end{corollary}
Two caveats, in full in \Cref{rem:Xhonest}: the witnesses are biased hypergraph states
and not hypergraph states, and the bias is chosen precisely so that $X$ does nothing, so
what \Cref{thm:Xcolumn} shows is that the bit-flip column has the same \emph{shape} as
the diagonal one --- not that $X$ interacts with the rotation hierarchy. Whether the
separations survive on families of modulus-$1$ amplitudes is open.

Adding $X$ to the group can be seen as the third dimension of the succinctness with exponential gain.
The reason is that on Boolean functions the diagonal groups buy nothing at all, but the $X$ gate can give exponential compression.

An $n$-variable Boolean function is a map $f:\{0,1\}^n\to\{0,1\}$. For $f$ not
identically $0$ we regard $f$ itself as the $n$-qubit vector $\sum_xf(x)\ket x$.\footnote{If the reader prefers, the function $f$ can be replaced by its normalization $\bar f=f(x)/\sqrt{|f^{-1}(1)|}$, where $|f^{-1}(1)|$ is the number of inputs $x$ that evaluate to $1$ under $f$. But this normalization does not change the statement and proof of \Cref{cor:X}.} 

\begin{restatable}{theorem}{obdd}\label{thm:obdd}
  Let $f:\{0,1\}^n\to\{0,1\}$ be a Boolean function. For every $k$ and $q$, a
  $\gen{C^{\leq k}R_q}$-LIMDD for $f$ under any variable order has at
  least as many nodes as the smallest OBDD for $f$ under that order.
\end{restatable}

The proof, in \Cref{sec:addingX}, uses only that diagonal maps preserve the support.

There are Boolean functions with exponential OBDD size and linear-size
tower-$\gen X$-LIMDDs, for instance the coset states of linear codes --- the
characteriztic functions of affine subspaces $V\subseteq\{0,1\}^n$, which are the
basis states of CSS codes~\cite{calderbank1997quantum,steane1996error}. By \cite[Theorem 10]{limdd} some such $V$ has OBDD size
$2^{\Omega(n)}/(2n)$ and a tower-$\gen X$-LIMDD for every variable order. Combined with
\Cref{thm:obdd} this gives the following succinctness result.

\begin{restatable}{corollary}{corX}\label{cor:X}
  For every $k$ and $q$, the class $\gen{C^{\leq k}R_q}$-LIMDD is not at least as
  succinct as $\gen X$-LIMDD. Consequently $\gen{C^{\leq k}R_q,X}$-LIMDD is
  exponentially more succinct than $\gen{C^{\leq k}R_q}$-LIMDD.
\end{restatable}

\subsubsection{Quantum Fourier transform is a tower.}
The hierarchy above is about states. Now we turn to operators, and we show that the 
quantum Fourier transform matrix has an efficient LIMDD representation.

We represent a matrix as a state in the usual way for decision diagrams. The
unprimed variable $x_\ell$ selects the top or bottom half of the matrix and the
primed variable $x'_\ell$ selects the left or right half, so that
$A=\left(\begin{smallmatrix}A_{x=0,x'=0}&A_{x=0,x'=1}\\A_{x=1,x'=0}&A_{x=1,x'=1}\end{smallmatrix}\right)$,
and the two variables of a level are adjacent in the order.
We show that the quantum Fourier transform can be represented as a Tower. This proves the experimental observations by~\cite{hong2025limtdd}.

\begin{theorem}[restate=thmqfttower]\label{thm:qft_is_tower}
  Every matrix representation of the quantum Fourier transform on $n$ qubits is a
  Tower-$\gen{R_N}$-LIMDD with $N=2^n$, for every variable order that keeps
  $x_\ell$ and $x'_\ell$ adjacent.
\end{theorem}

The proof is in \Cref{sec:qft}, together with the resulting diagram
(\Cref{fig:QFT-Tower-LIMDD}).

\begin{table}[b!]
  \centering
  \begin{tabular}{l|ccccc||cccccccc}
     & \rotatebox{90}{Sample} & \rotatebox{90}{Measure} & \rotatebox{90}{Equality} & \rotatebox{90}{InnerProd} & \rotatebox{90}{Fidelity} & \rotatebox{90}{$X$} & \rotatebox{90}{$Z$} & \rotatebox{90}{$T$} & \rotatebox{90}{$R_m$} & \rotatebox{90}{$C^{k'}R_m$} & \rotatebox{90}{$CX$} & \rotatebox{90}{$H$} & \rotatebox{90}{$\textit{SWAP}$} \\
    \hline
    $\gen{X,Z}$-LIMDD              & \Yes$^r$ & \Yes & \Yes & $\circ$   & $\circ$ & \Yes & \Yes & \Yes & \Yes & \Yes & \No & \No & \No \\
    $\gen{R_q}$-LIMDD              & \Yes$^r$ & \Yes & \Yes & $\circ^*$ & $\circ^*$ & \Yes & \Yes & \Yes & \Yes & \Yes & \No & \No & \No \\
    $\gen{C^{\leq k}R_q}$-LIMDD    & \Yes$^r$ & \Yes & \Yes & $\circ^*$ & $\circ^*$ & \Yes & \Yes & \Yes & \Yes & \Yes & \No & \No & \No \\
    $\gen{R_q,X}$-LIMDD            & \Yes$^r$ & \Yes & \Yes & $\circ^*$ & $\circ^*$ & \Yes & \Yes & \Yes & \Yes & \Yes & \No & \No & \No \\
    $\gen{C^{\leq k}R_q,X}$-LIMDD  & \Yes$^r$ & \Yes & \Yes & $\circ^*$ & $\circ^*$ & \Yes & \Yes & \Yes & \Yes & \Yes & \No & \No & \No
  \end{tabular}
  \caption{Query and gate tractability. The first row is partially reproduced
  from~\cite{vinkhuijzen2024a}. Parameters $m,q,k,k'$ are independent.\\
  \Yes: applicable in polytime. \Yes$^r$: randomised polytime. \No: not applicable in polytime. $\circ$: not polytime unless
  $\mathrm{P}=\mathrm{NP}$. $\circ^*$: $\circ$ if $2\mid q$, open otherwise.}
  \label{tab:tractability}
\end{table}

\subsubsection{Tractability.}

To motivate our new families of LIMDDs beyond succinctness considerations, we study the tractability of various operations and transformations. We consider the five queries described below, where the states $\ket\phi$ an $\ket\psi$ are represented as LIMDDs
\begin{itemize}
\item[•] \textbf{Sample}: return a random $x\in\{0,1\}^n$ with probability
$|\bra x\ket\phi|^2/\bra\phi\ket\phi$. 
\item[•] \textbf{Measure}: compute the probability $|\bra x\ket\phi|^2/\bra\phi\ket\phi$ for a given $x\in\{0,1\}^n$.
\item[•] \textbf{Equality}: determine whether $\ket\phi=\ket\psi$. 
\item[•] \textbf{InnerProd}: compute the inner product $\bra\phi\ket\psi$.
\item[•] \textbf{Fidelity}: compute $|\bra\phi\ket\psi|^2$.
\end{itemize}
We also study the tractability of applying relevant gates to LIMDDs. For $M$ a gate, we ask whether there is a polynomial-time algorithm which, given a LIMDD representation of state $\ket\phi$, returns a LIMDD of the same kind representing $M\ket\phi$. We consider eight (families of) gates: $X$, $Z$, $T$, $R_m$, $C^kR_m$, $CX$, $H$, and $SWAP$. Interestingly, we show that there is no notable difference of tractability (up to polynomial variations) between the five families of LIMDDs considered. For any gate listed above, the complexity of applying this gate is the same for $\gen{X,Z}$-LIMDD, $\gen{R_q}$-LIMDD, $\gen{C^{\leq k}R_q}$-LIMDD, $\gen{R_q,X}$-LIMDD  and $\gen{C^{\leq k}R_q,X}$-LIMDD. \Cref{sec:tractability} proves every entry of the table.

\subsubsection{Canonicity is polytime.}
We provide a normal form for the whole family of LIMDDs. We show that this normal form is unique, and that it can be calculated in polynomial time. More specific, a diagram irreducible
under the five rules R1--R5 below is the unique diagram for the state it
represents, for all $k$ and $q$, with and without $X$. The rule that
carries the coarsening is R4, which selects the canonical representative of an
orbit of high-edge labels, and identifying the orbit correctly is delicate: it is
strictly larger than the set of labels reachable by stabilizers of the two children,
by a free factor that the layer split of a group element produces (\Cref{prop:r4class}).

The rules refer to a total order $\preccurlyeq$ on nodes, arbitrary but fixed
once for all diagrams. \Cref{sec:rules} explains why a per-diagram topological
order would not suffice. Let $v$ be a node whose $0$-edge points to $v_0$ with label $L_0$ and whose
$1$-edge points to $v_1$ with label $L_1$. The five reduction rules are the
following:
\begin{itemize}
  \item[R1.] \emph{Zero edge.} If $L_b=0$ for some $b$, replace $v_b$ by $v_{1-b}$ (but still keeping $L_b=0$).
  \item[R2.] \emph{Low precedence.} If $X\in\gen G$, enforce $v_0\preccurlyeq v_1$.
  \item[R3.] \emph{Low factoring.} If $X\in\gen G$, enforce $L_0=I^{\otimes n}$.
    Otherwise enforce that $L_0=I^{\otimes n}$, or that $L_0=0$ and
    $L_1=I^{\otimes n}$.
  \item[R4.] \emph{High determinism.} Replace $L_1$ by the $\leq_{\mathrm{lex}}$-least
    element of the set of labels that keep $v$ in its $\gen G$-equivalence class.
    Here $\leq_{\mathrm{lex}}$ is the lexicographic order on \emph{labels} fixed in
    \eqref{eq:order}; it is not the order $\preccurlyeq$ on nodes used by R2.
    \Cref{prop:r4class} identifies that set.
  \item[R5.] \emph{Merge.} If a node $u$ has the same children and the same edge
    labels as $v$, merge $u$ and $v$.
\end{itemize}

\Cref{fig:reduced1} shows the node these rules produce, in the style
of~\cite{limdd}. Each $\rightsquigarrow$ is labeled with the rule that fires, and
whatever a rule removes from an edge reappears on the root edge.

\begin{figure}[b!]
\centering
\tikz[->,>=stealth',shorten >=1pt,auto,node distance=1.5cm,font=\footnotesize,
        thick, state/.style={circle,draw,inner sep=0pt,minimum size=14pt}]{
    \node[state] (1) {$v'$};
    \node[state] (1a) [below = 1cm of 1, xshift=1.7cm] {$w$};
    \node[above = .5cm of 1] (x1) {};
    \path[]
    (x1) edge      node[pos=.4,above right,pos=.7] {} (1)
    (1) edge[e0,bend left=-20] node[pos=.3,lbl,left] {$A$} (1a)
    (1) edge[e1,bend right=-20] node[pos=.3,lbl,above right] {$0$} (1a)
    ;
    \node[state, right = 2.5cm of 1] (2) {$v$};
    \node[above = .5cm of 2] (x2) {};
    \path[]
    (x2) edge    node[left,pos=0,lbl] {$\id \otimes A$} (2)
    (2) edge[e0,bend left=-20] node[pos=.3,left] {} (1a)
    (2) edge[e1,bend right=-20] node[pos=.3,lbl,right] {$0$} (1a)
    (1) --  node[yshift=.0cm] {$\overset{\mathrm{R3}}{\rightsquigarrow}$} (2)
    ;
    }~~~~~~~~~~~~~~~
\tikz[->,>=stealth',shorten >=1pt,auto,node distance=1.5cm,font=\footnotesize,
        thick, state/.style={circle,draw,inner sep=0pt,minimum size=14pt}]{
    \node[state] (1) {$v''$};
    \node[above = .5cm of 1] (x1) {};
    \node[state] (1a) [below = 1cm of 1, xshift=2.3cm] {$v_L$};
    \node[state] (1b) [below = 1cm of 1, xshift=3.8cm] {$v_R$};
    \path[]
    (1) edge[e0] node[pos=.37,lbl,left] {$A$} (1a)
    (1) edge[e1] node[pos=.2,lbl,above right] {$B$} (1b)
    (1a) --  node[yshift=-.2cm] {$\beforeq$} (1b)
    ;
    \node[state, right = 2.5cm of 1] (2) {$v'$};
    \node[above = .5cm of 2] (x2) {};
    \path[]
    (x1) edge     node[above left,pos=.4] {} (1)
    (x2) edge     node[left,lbl,pos=.0] {$\id \otimes A$} (2)
    (2) edge[e0] node[pos=.2,left] {} (1a)
    (2) edge[e1] node[pos=.16,lbl,right] {$ A^{-1}B$} (1b)
    (1) --  node[yshift=.0cm] {$\overset{\mathrm{R3}}{\rightsquigarrow}$} (2)
    ;
    \node[state, right = 2.5cm of 2] (3) {$v$};
    \node[above = .5cm of 3] (x3) {};
    \path[]
    (x3) edge    node[lbl,above left,pos=.4] {$(\id \otimes A)\rootlim$} (3)
    (3) edge[e0] node[pos=.2,above left] {} (1a)
    (3) edge[e1] node[pos=.3,lbl,below right] {$\highlim$} (1b)
    (2) --  node[yshift=.0cm] {$\overset{\mathrm{R4}}{\rightsquigarrow}$} (3)
    ;
    }
	\caption{Reduced node construction in case $\ket{\phi_1} = 0$ (left), and
	        $\ket{\phi_0}, \ket{\phi_1} \neq 0$ and $v_L \beforeq v_R$ (right).
	        Each $\rightsquigarrow$ carries the rule that fires; R1 has already
	        redirected the zero edge on the left, and R2 supplies $v_L \beforeq v_R$
	        on the right.
	        Not shown: for cases $\ket{\phi_0} = 0$ and $v_R \beforeq v_L$, we take instead root edge $ X \otimes A$ and swap low/high edges. R5 is applied by hashing the identifier of node $v$.}
	\label{fig:reduced1}
\end{figure}

The set that R4 minimizes over is stated in the following proposition, in terms
of the stabilizers $Stab(v_0),Stab(v_1)$ of \Cref{sec:prelim}.

\begin{proposition}[restate=propfour,name={the R4 class}]\label{prop:r4class}
  Let $G=\gen{C^{\leq k}R_q}$ or $G=\gen{C^{\leq k}R_q,X}$, and let $v$ be a node
  with $L_0=I^{\otimes n}$, children $v_0,v_1$ and high label $L_1$. Then
  \begin{equation}\label{eq:true_label_class}
    \left\{L_1'\ \middle|\ \ket0\otimes\ket{v_0}+\ket1\otimes L_1'\ket{v_1}
      \ \text{is}\ \gen G\text{-equivalent to}\ \ket v\right\}
    =\gen{C^{\leq k-1}R_q}\cdot Stab(v_0)\cdot L_1\cdot Stab(v_1)
  \end{equation}
  when $X\notin\gen G$ or the children denote states that are not
  $\gen G$-equivalent --- for canonical children, exactly when $v_0\neq v_1$.
  Otherwise the set additionally contains the \emph{inverse-twisted} coset
  $\gen{C^{\leq k-1}R_q}\cdot Stab(v_0)\cdot (L_1)^{-1}\cdot Stab(v_1)$.
\end{proposition}

\begin{theorem}[restate=thmcanonicity,name={}]\label{thm:canonicity}
  Let $G=\gen{C^{\leq k}R_q}$ or $G=\gen{C^{\leq k}R_q,X}$. Every $\gen G$-LIMDD
  that satisfies R1--R5 at every node is in canonical form. That is, for every
  nonzero vector $\ket\phi$ there is a \emph{unique} node $v$ satisfying R1--R5
  such that $\ket\phi=\ell\ket v$ for some $G$-LIM $\ell$.
\end{theorem}

\Cref{sec:canonform} gives the proof idea and \Cref{app:canon} the full proof by
induction.

\begin{theorem}[restate=thmpolytime]\label{thm:polytime}
  Applying R1--R5 to a $\gen G$-LIMDD can be done in time polynomial in the size of
  the LIMDD, for $\gen G=\gen{C^{\leq k}R_q}$ and for $\gen G=\gen{C^{\leq k}R_q,X}$,
  for every $q\geq1$.
\end{theorem}

In the diagonal case the orbit of R4 is a coset of a computable subgroup of
$\gen{C^{\leq k}R_q}$, which is in fact a subgroup of $\z_q^M$ once labels are written in the
exponent coordinates of \Cref{sec:machinery}, with $M\in O(n^{k+1})$. The
canonical representative is the lexicographic minimum of the coset, computed by Howell
reduction. With the bit flip $X$, the orbit is also a coset, but $\gen{C^{\leq k}R_q,X}$ is no cyclic group and therefore does not allow representations of subgroups in Howell normal form. Fortunately, $\gen{C^{\leq k}R_q,X}$ is a semi-direct product group of $\gen{C^{\leq k}R_q}$ and $\gen{X}$, two groups that ensure a subgroup representation in Howell normal form. This semi-direct product group structure allows efficient computation and representation of subgroups of $\gen{C^{\leq k}R_q,X}$, as well as efficient computation of a canonical representative of a coset.

Throughout, ``polynomial in the size of the LIMDD'' means polynomial for each
fixed $k$: a single edge label already has
$M=\sum_{\ell=1}^{k+1}\binom{n}{\ell}\in O(n^{k+1})$ digits
(\Cref{sec:machinery}), so the diagram size and the time to calculate the canonical form both scale exponentially
in $k$.

\subsubsection{The locality boundary.}
Finally we show where coarsening has to stop. Every group in
the family consists of diagonal and antidiagonal matrices, and that is not an
accident of the analysis. If $G$ contains a matrix that is neither, then two nodes
that are not in one orbit can become redundant once a node is placed above them, so
orbit equality is no longer decided from the children and the quotient stops being a
bottom-up congruence closure. We exhibit this situation for $G=\gen H$, where $H$ is the Hadamard gate.

\begin{example}\label{ex:H}
  Take the unnormalized one-qubit states
  $\ket\phi=(\tfrac1{\sqrt2}+1)\ket0+\tfrac1{\sqrt2}\ket1$ and
  $\ket\psi=(\tfrac1{\sqrt2}-1)\ket0+\tfrac1{\sqrt2}\ket1$. Neither they nor their
  normalizations are related by an $\gen H$-LIM, so the two nodes cannot be merged.
  Yet
  \begin{equation*}
    \ket\Phi=\tfrac1{\sqrt2}\bigl(\ket0\otimes\ket\phi+\ket1\otimes\ket\psi\bigr)
    =(H\otimes H)\bigl[(\ket0+H\ket1)\otimes\ket0\bigr],
  \end{equation*}
  which is a Tower-$\gen H$-LIMDD: both levels have a single node. So the two
  unmergeable nodes disappear once a node is placed above them.%
\end{example}

\Cref{ex:H} is the mirror image of \Cref{lem:arity}. For arity-closed groups of
(anti-)diagonal matrices, inequivalent children stay inequivalent when a level is
added, so a bottom-up construction can decide sharing from local data. For
$\gen H$ the opposite happens: two nodes that a local algorithm has already
declared distinct become redundant later. An $\gen H$-LIMDD algorithm therefore
either rebuilds the whole diagram on each edge creation, and is not local, or
returns a diagram that is not the smallest one. 

\subsection*{Organization.}
\Cref{sec:succinctness} proves the succinctness order (\Cref{fig:succinctness_picture}); \Cref{sec:tractability}
the tractability table (\Cref{tab:tractability}); \Cref{sec:canonicity} introduces the
the exponent-vector encoding and the layer split, the one calculation the rest of
the paper reuses, and proves \Cref{thm:canonicity} and \Cref{thm:polytime}; \Cref{sec:qft} proves the the quantum Fourier transform matrix as tower. Proofs that are routine,
long, or purely computational are in the appendix.

\FloatBarrier

\section{A succinctness hierarchy}\label{sec:succinctness}

The hierarchy of \Cref{fig:succinctness_picture} rests on three rules, all for a
fixed second parameter: coprime $q$ and $p$ give incomparable classes of $\gen{C^{\leq k}R_q}$-LIMDD and $\gen{C^{\leq k}R_p}$-LIMDD; $p\mid q$
with $q>p$ makes $\gen{C^{\leq k}R_q}$-LIMDD exponentially more succinct than
$\gen{C^{\leq k}R_p}$-LIMDD; and $h<k$ makes $\gen{C^{\leq k}R_q}$-LIMDD
exponentially more succinct than $\gen{C^{\leq h}R_q}$-LIMDD. All three come from
one family of separating states, built from hypergraphs. \Cref{sec:towers} shows
the family has linear-size towers, \Cref{sec:lowerbound} shows it has no small
diagram in the weaker classes, and \Cref{sec:addingX} adds the third direction, the
bit flip.

\subsection{Hypergraph states are tower LIMDDs}\label{sec:towers}

We let the reader recall \Cref{def:hypergraphstate} for hypergraph states. We introduce the unnormalized state
$$
\ket{\hat\phi^G_p}:= 2^{m/2}\ket{\phi^G_p} = \prod_{e\in E}(C^{|e|-1}R_p)_{(e)}(\ket0+\ket1)^{\otimes m}.
$$
In our proofs, we construct LIMDDs for $\ket{\hat\phi^G_p}$. Adding a factor $2^{-m/2}$ to the root edge thus gives LIMDDs for $\ket{\phi^G_p}$ . 

\begin{figure}[t]
    \centering
    \begin{tikzpicture}
    [
        he/.style={draw, semithick},        %
    ]
    
    \node (1) at (0,0) {$5$};
    \node (2) at (1,1) {$4$};
    \node (3) at (1,0) {$3$};
    \node (4) at (1,2) {$2$};
    \node (5) at (0,2) {$1$};
    
    \draw[he] \hedgeii{4}{5}{3mm};
    \draw[he] \hedgeii{1}{3}{3mm};
    \draw[he] \hedgeii{1}{3}{4mm};
    \draw[he] \hedgeiii{2}{5}{4}{4mm};
    \draw[he] \hedgeiii{2}{5}{4}{5mm};
    \draw[he] \hedgeiii{1}{4}{3}{6mm};
    \end{tikzpicture}~~~~
    \begin{tikzpicture}[
        scale=0.3,
        every path/.style={>=latex},
        every node/.style={},
        inner sep=1pt,
        minimum size=0.3cm,
        line width=1pt,
        node distance=.5cm,
        thick,
        font=\scriptsize
        ]
        \node[] (a0) {};
        \node[draw,circle,below = of a0] (a1) {$5$};
        \node[draw,circle,below = of a1] (a2) {$4$};
        \node[draw,circle,below = of a2] (a3) {$3$};
        \node[draw,circle,below = of a3] (a4) {$2$};
        \node[draw,circle,below = of a4] (a5) {$1$};
        \node[draw,circle,below = of a5] (a6) {$0$};

        \draw[e1] (a0) edge  node[right] {$I^{\otimes 5}$} (a1);
        
        \draw[e0=20] (a1) edge node[left,pos=.3] {} (a2);
        \draw[e1=20] (a1) edge  node[right,pos=.3] {~$(C^2R_q)_{(2,3,4)}(R_q)^2_{(3)}$} (a2);

        \draw[e0=20] (a2) edge node[left,pos=.3] {} (a3);
        \draw[e1=20] (a2) edge  node[right,pos=.3] {~$(CR_q)^2_{(1,2)}$} (a3);

        \draw[e0=20] (a3) edge node[left,pos=.3] {} (a4);
        \draw[e1=20] (a3) edge  node[right,pos=.3] {~$I^{\otimes 2}$} (a4);

        \draw[e0=20] (a4) edge node[left,pos=.3] {} (a5);
        \draw[e1=20] (a4) edge  node[right,pos=.3] {~$R_q$} (a5);

        \draw[e0=20] (a5) edge node[left,pos=.3] {} (a6);
        \draw[e1=20] (a5) edge  node[right,pos=.3] {~$1$} (a6);
    \end{tikzpicture}
    \caption{A $4$-hypergraph $G$ and its hypergraph state $\ket{\phi^G_q}$ as a Tower-$\gen{C^{\leq2}R_q}$-LIMDD, for any $q$. The $1$-edge label of the node for vertex $j$ collects exactly the hyperedges whose largest vertex is $j$, with that vertex removed. Reading the labels off the tower inverts the construction, which is \Cref{thm:tower_is_hypergraphstate}.}
    \label{fig:hypergraph}
\end{figure}
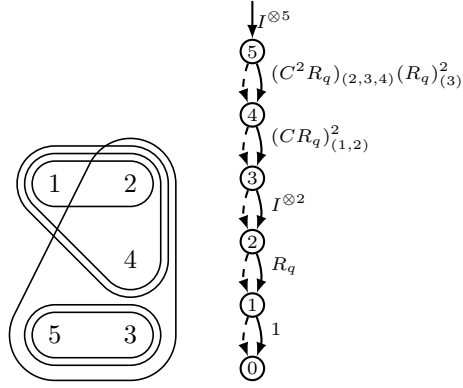

\thmhgtower*
\begin{proof}
  Let $\sigma=(x_1,\dots,x_m)$ be a variable order and induct on $m$. For $m=1$ the
  claim is immediate. For $m>1$, put $E^0=\{e\in E\mid x_1\notin e\}$ and
  $E^1=\{e\setminus\{x_1\}\mid e\in E,\ x_1\in e\}$, write $x_e=\prod_{x_i\in e}x_i$
  and $G'=(\{x_2,\dots,x_m\},E^0)$. Splitting the sum on $x_1$,
  \begin{align*}
    \ket{\hat\phi^G_p}
      &=\ket0\otimes\!\!\sum_{x\in\{0,1\}^{m-1}}\prod_{e\in E^0}\omega_p^{x_e}\ket x
       +\ket1\otimes\!\!\sum_{x\in\{0,1\}^{m-1}}\prod_{e\in E^1}\omega_p^{x_e}\prod_{e\in E^0}\omega_p^{x_e}\ket x\\
      &=\ket0\otimes\ket{\hat\phi^{G'}_p}
       +\ket1\otimes\Bigl[\textstyle\prod_{e\in E^1}(C^{|e|-1}R_p)_{(e)}\Bigr]\ket{\hat\phi^{G'}_p} .
  \end{align*}
  The two cofactors are equal up to the $\gen{C^{\leq k-2}R_p}$-LIM
  $\prod_{e\in E^1}(C^{|e|-1}R_p)_{(e)}$, which lies in the group because $|e|\leq k$
  and hence $|e\setminus\{x_1\}|\leq k-1$. So the two edges of the node for $x_1$
  point to the same child, and the induction continues on $\ket{\hat\phi^{G'}_p}$
  with $m-1$ vertices. Putting $2^{-m/2}$ on the root edge gives $\ket{\phi^G_p}$.
\end{proof}

\subsection{An exponential lower bound}\label{sec:lowerbound}

Let $G=(V,E)$ be a graph and $V_1,V_2\subseteq V$ disjoint. A
$(V_1,V_2)$-\emph{induced matching} is a set $M\subseteq E\cap(V_1\times V_2)$ of
pairwise disjoint edges such that $G[V(M)]=M$, that is, no two edges of $M$ are
joined by a further edge of $G$. The lower bound needs one combinatorial fact about
grids, which we state separately because it contains the core argument.

Each node $v$ of the graph $G$ contributes $k$ nodes $v_1,\dots,v_k$ in the $k$-hyperization of $G$. 
We call $v$ a \emph{supernode} and the $v_1,\dots,v_k$ \emph{subnodes} of $v$. We denote by $\zeta$ the function that
maps subnodes to their corresponding supernode.

\begin{lemma}[{\cite[Lemma A.2]{vinkhuijzen2025thesis}}]\label{lem:matching}
  Let $T\subseteq V(G_n')$ with $|T|\geq|V|/2$ and $|\bar T|\geq|V|/2$. Then
  $G_n'$ has a $(T,\bar T)$-induced matching of size at least $\lfloor n/12\rfloor$.
\end{lemma}

\thmprimehierarchy*
\begin{proof}

  Let $\sigma=(x_1,\dots,x_{(k+1)n^2})$ be a variable order. Let $d$ be least such that
  half of the supernodes have at least one subnode among $x_1,\dots,x_d$, put
  $X_d=\{x_1,\dots,x_d\}$, and let $T'$ be the set of supernodes with a subnode in
  $X_d$. By \Cref{lem:matching} there is a $(T',\bar T')$-induced matching $M$ in
  $G_n'$ of size at least $\lfloor n/12\rfloor$. Let $U'\subseteq T'$ be the set of
  its endpoints in $T'$.

  Pick two distinct $A',B'\subseteq U'$. For each supernode in $A'$ choose one of
  its subnodes lying in $X_d$, and let $A$ be the set of chosen subnodes; define
  $B$ from $B'$ likewise. Write
  \begin{equation}\label{eq:f_expression}
    f(x)=\prod_{e\in E}\omega_p^{\prod_{x_i\in e}x_i},
    \qquad\text{so}\qquad
    \ket{\psi^n_{k+1,p}}=2^{-m/2}\sum_{x}f(x)\ket x .
  \end{equation}
  Let $\vec\alpha$ be the assignment to $X_d$ that sets $A$ to $1$ and everything
  else to $0$, and $\vec\beta$ likewise from $B$. We show the cofactors
  $f_{\vec\alpha}$ and $f_{\vec\beta}$ are not $\gen{C^{\leq h}R_q}$-equivalent.
  Suppose $f_{\vec\alpha}=\pi\cdot f_{\vec\beta}$ for a LIM
  \begin{equation}\label{eq:pi_expression}
    \pi=\lambda\cdot\prod_{0\leq \ell\leq h+1}\ \prod_{S\in\binom{[m-d]}{\ell}}\omega_q^{d_Sx_S},
    \qquad \lambda\in\co,\ d_S\in\z_q ,
  \end{equation}
  where $x_S=\prod_{i\in S}x_i$. The range of $\ell$ is the arity bound: the
  generator $C^{h'}R_q$ has arity $h'+1$, so the monomials available to $\pi$ are
  those of degree at most $h+1$.

  Choose $c'\in A'\setminus B'$; such $c'$ exists after possibly swapping $A'$ and
  $B'$, since $A'\neq B'$. Let $c$ be its chosen subnode, let $j'\in\bar U'$ be the
  supernode matched to $c'$ by $M$, and let $J=\zeta^{-1}(j')$, so $|J|=k+1$. For
  $Z\subseteq J$ let $\mu_Z$ be the assignment to the remaining $m-d$ variables that
  is $1$ exactly on $Z$. Because $M$ is \emph{induced}, no edge of $G^k_n$ other
  than those between $\zeta^{-1}(c')$ and $J$ can be made true by $\mu_Z$ together
  with $\vec\alpha$ or $\vec\beta$. Hence
  \begin{equation}\label{eq:pi_is_one}
    f_{\vec\alpha}(\mu_Z)=f_{\vec\beta}(\mu_Z)=1
    \quad\text{for }|Z|<k+1,
    \qquad\text{so}\qquad \pi(\mu_Z)=1 ,
  \end{equation}
  while for $Z=J$ the edge $\{c\}\cup J$ becomes true under $\vec\alpha$ but not
  under $\vec\beta$, so
  \begin{equation}\label{eq:pi_is_omega}
    f_{\vec\beta}(\mu_J)=1,\quad f_{\vec\alpha}(\mu_J)=\omega_p,
    \qquad\text{so}\qquad \pi(\mu_J)=\omega_p .
  \end{equation}
  Taking $Z=\emptyset$ in \eqref{eq:pi_is_one} gives $\lambda\,\omega_q^{d_\emptyset}=1$.
  Only this product is used below --- $\mu_\emptyset$ is the all-zero assignment, so
  it is one equation, not two --- and $\omega_q^{d_\emptyset}$ is a scalar, so we may
  absorb it into $\lambda$ and assume $d_\emptyset=0$ and $\lambda=1$. Now induct on $|Z|$: if $d_{Z'}=0$ for all $Z'\subsetneq Z$ with
  $|Z'|<|Z|$, then every factor of \eqref{eq:pi_expression} except $\omega_q^{d_Z}$
  evaluates to $1$ at $\mu_Z$, because $x_S(\mu_Z)=0$ unless $S\subseteq Z$. So
  \eqref{eq:pi_is_one} gives $\omega_q^{d_Z}=1$, that is $d_Z=0$. This holds for
  every $Z\subseteq J$ with $|Z|<k+1$.

  If $h<k$ then \eqref{eq:pi_expression} has no component of arity $k+1$, so $\pi$ is
  $1$ on every $\mu_Z$ with $Z\subseteq J$, contradicting \eqref{eq:pi_is_omega}. If
  $h\geq k$ then $\pi(\mu_J)=\omega_q^{d_J}$, and \eqref{eq:pi_is_omega} requires
  $\omega_q^{d_J}=\omega_p$, which is impossible when $q$ is not a multiple of $p$.
  So $\pi$ does not exist.

  Since $A'$ and $B'$ were arbitrary distinct subsets of $U'$, the diagram has at
  least $2^{|U'|}\geq2^{\lfloor n/12\rfloor}$ pairwise non-equivalent nodes at
  level $d$. By \Cref{lem:arity} these nodes stay pairwise non-equivalent in any
  larger diagram, so none of them can be merged.
\end{proof}

Together with \Cref{thm:hypergraphstate_is_tower} this gives the three rules above,
hence \Cref{fig:succinctness_picture}. \Cref{cor:continuous} also follows from \Cref{theorem:prime-hierarchy}.
\corContinuous* 
\begin{proof}
  Every $\gen{C^{\leq k}R_q}$-LIMDD is a $C^{\leq k}\Rot$-LIMDD, so
  $C^{\leq k}\Rot$-LIMDD is at least as succinct. Conversely $C^{\leq k}\Rot$
  contains $\gen{C^{\leq k}R_p}$ for every $p$ coprime to $q$, so
  \Cref{theorem:prime-hierarchy} with that $p$ separates them.
\end{proof}

\subsection{Adding the bit flip}\label{sec:addingX}

We first prove that allowing the bit flip $X$ in the labels of the LIMDDs give exponentially more succinct representations.

\obdd*
\begin{proof}
  It suffices to show, for every prefix length $j$ and all assignments
  $\alpha,\beta$ to the first $j$ variables, that $f_\alpha\neq f_\beta$ implies
  $f_\alpha\neq\pi f_\beta$ for every $\gen{C^{\leq k}R_q}$-LIM $\pi$. Let
  $\gamma$ witness $f_\alpha(\gamma)\neq f_\beta(\gamma)$, say
  $f_\beta(\gamma)=0$. Every element of $\gen{C^{\leq k}R_q}$ is diagonal with
  nonzero diagonal, so it preserves the set of inputs on which a function vanishes.
  Hence $\pi f_\beta$ vanishes at $\gamma$, while $f_\alpha(\gamma)\neq0$.
  So the two prefixes reach distinct nodes.
\end{proof}

\corX*
\begin{proof}
The proof follows directly from \Cref{thm:obdd} and \cite[Theorem 10]{limdd}. 
By \cite[Theorem 10]{limdd} we know that there exists a family of Boolean functions $f:\{0,1\}^n\to\{0,1\}$ such that with high probability the representation of $f$ has OBDD size
$2^{\Omega(n)}/(2n)$ and a tower-$\gen X$-LIMDD for every variable order. Combined with
\Cref{thm:obdd} this gives corollary directly.
\end{proof}

We now prove the succinctness hierarchy for LIMDDs equipped with the bit flip, as formalized in \Cref{thm:Xcolumn} and \Cref{cor:Xhierarchy}. 
For that we use the biased hypergraph states (\Cref{def:biasedhypergraphstate}). 
Three lemmas do the work: the bias does not impact succinct tower-LIMDD representations (\Cref{lem:biasedtower}), the bias kills the bit flip
(\Cref{lem:rigidity}), and what is then left is a diagonal comparison, which we record in
coefficient rather than evaluation form (\Cref{lem:exponentform}). 

\begin{lemma}\label{lem:biasedtower}
  Let $k,p\in\natur$, let $G$ be a $(k+2)$-hypergraph and let $r\neq0$. Then
  $\ket{\phi^G_{p,r}}$ is represented by a Tower-$\gen{C^{\leq k}R_p}$-LIMDD, for every
  variable order.
\end{lemma}
\begin{proof}
  Run the induction of \Cref{thm:hypergraphstate_is_tower} on
  $(\ket0+r\ket1)^{\otimes m}$ instead of $\ket+^{\otimes m}$. Splitting on $x_1$, every
  $x$ in the high branch has one more $1$ than its low counterpart, so in the notation of
  that proof the high cofactor is
  $r\cdot\bigl[\prod_{e\in E^1}(C^{|e|-1}R_p)_{(e)}\bigr]$ applied to the low one. The
  arity is at most $k$ exactly as there, and $r\cdot g$ with $g\in\gen{C^{\leq k}R_p}$ is
  a $\gen{C^{\leq k}R_p}$-LIM, so the two edges of the node for $x_1$ again point to the
  same child.
\end{proof}

\begin{lemma}[rigidity]\label{lem:rigidity}
  Let $\ket\psi,\ket{\psi'}$ be $m$-qubit states with $|\psi(y)|=c\,t^{|y|}$ and
  $|\psi'(y)|=c'\,t^{|y|}$ for some $t>0$, $t\neq1$ and $c,c'>0$. If
  $\ket\psi=\pi\ket{\psi'}$ for a $\gen{C^{\leq h}R_q,X}$-LIM $\pi$, then $\pi$ is
  diagonal.
\end{lemma}
\begin{proof}
  Throughout this proof $S$ denotes a set of qubits and $a\in\{0,1\}^m$ a bit string;
  we write $X^{(a)}=\prod_{j=1}^{m}\left[X_{(j)}\right]^{a_j}$ for the flip of the
  qubits on which $a$ is $1$, so that $X^{(a)}\ket y=\ket{y\oplus a}$ with $y\oplus a$
  the string $y$ with the bits of $a$ flipped, and $|a|$ for the Hamming weight of $a$.
  By \Cref{lem:commute} the bit flips of $\pi$ can be moved to the right, where they
  collect into a single $X^{(a)}$ as each $X_{(j)}$ is an involution and flips on
  distinct qubits commute; the factors before them are diagonal and commute with one
  another, so
  \begin{equation}\label{eq:rigidity_form}
    \pi=\lambda\cdot\operatorname{diag}(\omega_q^{f(y)})\cdot X^{(a)},
    \qquad \lambda\in\co\setminus\{0\},\ a\in\{0,1\}^m ,
  \end{equation}
  where $\operatorname{diag}(\omega_q^{f(y)})$ carries $\omega_q^{f(y)}$ at $y$ and
  $f=\sum_{S\subseteq[m]}d_Sx_S$ is multilinear over $\z_q$, with
  $x_S=\prod_{i\in S}x_i$ as in \eqref{eq:pi_expression}, and $x_\emptyset=1$: the
  factor $\left[(C^{|S|-1}R_q)_{(S)}\right]^{d_S}$ multiplies $\ket y$ by
  $\omega_q^{d_Sx_S(y)}$, so all of them together multiply it by $\omega_q^{f(y)}$.
  Each generator $C^{h'}R_q$ with $h'\leq h$ has arity $h'+1$, and \eqref{eq:commute}
  turns a flip past a rotation of arity $\ell$ into rotations of arity $\ell$ and
  $\ell-1$, so moving the flips rightwards never raises an arity: $d_S=0$ once
  $|S|>h+1$. Applying \eqref{eq:rigidity_form} to $\ket{\psi'}$ and reading off the
  amplitude at $y$ gives $\psi(y)=\lambda\,\omega_q^{f(y)}\,\psi'(y\oplus a)$, as
  $\oplus\,a$ is its own inverse. The diagonal entries have modulus $1$ --- the only
  property of them used below --- so in moduli
  $c\,t^{|y|}=|\lambda|\,c'\,t^{|y\oplus a|}$ for every $y$.

  Suppose $a\neq0$, say $a_i=1$. At $y=0$ we have $y\oplus a=a$, so this reads
  $c=|\lambda|\,c'\,t^{|a|}$. At the $y$ that is $1$ in position $i$ and $0$ elsewhere,
  $y\oplus a$ is $a$ with its $i$-th bit cleared, so $|y|=1$ and $|y\oplus a|=|a|-1$,
  and it reads $c\,t=|\lambda|\,c'\,t^{|a|-1}$. Multiplying the second equation by $t$
  and substituting the first gives $c\,t^{2}=|\lambda|\,c'\,t^{|a|}=c$, so $t^{2}=1$ as
  $c>0$, and $t=1$ as $t>0$, contradicting $t\neq1$. Hence $a=0$, so $X^{(a)}$ is the
  identity and $\pi=\lambda\cdot\operatorname{diag}(\omega_q^{f(y)})$ is diagonal.
\end{proof}

\begin{lemma}[exponent form]\label{lem:exponentform}
  Write $x_S=\prod_{i\in S}x_i$, let $g=\sum_Sg_Sx_S$ and $g'=\sum_Sg'_Sx_S$ be
  multilinear over $\z_p$, let $r\neq0$ with $|r|\neq1$, and put
  $\psi(y)=c\,r^{|y|}\omega_p^{g(y)}$, $\psi'(y)=c'r^{|y|}\omega_p^{g'(y)}$ with
  $c,c'\neq0$. If $\ket\psi=\pi\ket{\psi'}$ for a $\gen{C^{\leq h}R_q,X}$-LIM $\pi$ with
  exponents $(d_S)_{|S|\leq h+1}$ over $\z_q$, then for every $S\neq\emptyset$
  \begin{enumerate}
    \item $q\,(g_S-g'_S)=0$ in $\z_p$, and
    \item if $q=p$ then $g_S-g'_S=d_S$; in particular $g_S=g'_S$ when $|S|>h+1$.
  \end{enumerate}
  \leanref{Paper.lean\#L815}{succ\_exponent\_root}\;
  \leanref{Paper.lean\#L809}{succ\_exponent\_arity}\;
  \leanref{Paper.lean\#L793}{succ\_zeta\_const}
\end{lemma}
\begin{proof}
  Write $m$ for the number of qubits, so that $y$ ranges over $\{0,1\}^m$ and $S$ over
  the subsets of $[m]$. Since $|\omega_p^{g(y)}|=1$, the moduli are
  $|\psi(y)|=|c|\,t^{|y|}$ and $|\psi'(y)|=|c'|\,t^{|y|}$ with $t=|r|$, which is
  positive because $r\neq0$ and differs from $1$ by hypothesis. So \Cref{lem:rigidity}
  applies: $\pi$ is diagonal, that is, $a=0$ in \eqref{eq:rigidity_form}, leaving
  $\pi=\lambda\cdot\operatorname{diag}(\omega_q^{f(y)})$, with $f=\sum_Sd_Sx_S$ over
  $\z_q$ the exponents of the statement and $d_S=0$ once $|S|>h+1$.
  Comparing amplitudes in $\ket\psi=\pi\ket{\psi'}$, cancelling the common factor
  $r^{|y|}\neq0$ and dividing by $c$ and by $\omega_p^{g'(y)}$ gives
  \begin{equation}\label{eq:exponent_pointwise}
    \omega_p^{g(y)-g'(y)}=\Lambda\cdot\omega_q^{f(y)}\quad\text{for all }y,
    \qquad\text{where }\Lambda=\lambda c'/c .
  \end{equation}

  We now claim that if $\delta=\sum_S\delta_Sx_S$ is multilinear over
  $\z_p$ and $\omega_p^{\delta(y)}$ does not depend on $y$, then $\delta_S=0$ for every
  $S\neq\emptyset$. To see it, let $\mu_Z$ for $Z\subseteq[m]$ be the assignment that is
  $1$ exactly on $Z$, as in the proof of \Cref{theorem:prime-hierarchy}. Every $y$ is
  such a $\mu_Z$, and $x_S(\mu_Z)=1$ if $S\subseteq Z$ and $0$ otherwise, the empty
  product being $x_\emptyset=1$, so $\delta(\mu_Z)=\sum_{S\subseteq Z}\delta_S$. As
  $\omega_p$ is a primitive $p$-th root of unity, $\omega_p^u=\omega_p^v$ forces $u=v$
  in $\z_p$, so the hypothesis reads
  $\delta(\mu_Z)=\delta(\mu_\emptyset)=\delta_\emptyset$ for every $Z$, that is,
  $\sum_{\emptyset\neq S\subseteq Z}\delta_S=0$. Now induct on $|Z|$ as in that proof:
  for $Z\neq\emptyset$ every nonempty $S\subsetneq Z$ is smaller and contributes $0$, so
  the sum collapses to $\delta_Z$. Only $\delta_\emptyset$ is left free --- it is the
  common value itself --- which is why both conclusions exclude $S=\emptyset$.

  Raise \eqref{eq:exponent_pointwise} to the $q$-th power. This kills the right-hand
  phase, as $q\,f(y)=0$ in $\z_q$, and leaves $\omega_p^{q(g(y)-g'(y))}=\Lambda^q$ for
  all $y$. The polynomial $q\,(g-g')$ is multilinear over $\z_p$ with coefficients
  $q\,(g_S-g'_S)$, so the fact gives conclusion (1). For conclusion (2) we have
  $\omega_q=\omega_p$ and $d_S\in\z_p$, so \eqref{eq:exponent_pointwise} reads
  $\omega_p^{g(y)-g'(y)-f(y)}=\Lambda$, and the fact applies to $g-g'-f$, whose
  coefficients are $g_S-g'_S-d_S$. Hence $g_S-g'_S=d_S$ for every $S\neq\emptyset$, and
  in particular $g_S=g'_S$ when $|S|>h+1$, where $d_S=0$.
\end{proof}

\thmxcolumn*
\begin{proof}
  Fix a variable order and take $d$, $X_d$, $T'$, $M$ and $U'$ from the proof of
  \Cref{theorem:prime-hierarchy}, so $|U'|\geq\lfloor n/12\rfloor$. For $A'\subseteq U'$
  let $\psi_{A'}$ be the cofactor of $\ket{\psi^n_{k+1,p,r}}$ at the assignment to $X_d$
  that sets one chosen subnode of each supernode of $A'$ to $1$ and everything else to
  $0$. Its amplitude moduli are $|r|^{|A'|}|r|^{|y|}$, so \Cref{lem:rigidity} applies to
  every pair and \Cref{lem:exponentform} governs the comparison.

  Let $A'\neq B'$, pick $c'\in A'\setminus B'$, let $j'$ be its partner under $M$ and put
  $J=\zeta^{-1}(j')$, so $|J|=k+1$ and $J\cap X_d=\emptyset$. The coefficient of $x_J$ in
  the phase polynomial of $\psi_{A'}$ counts the hyperedges $\{v\}\cup J$ with $v$ set to
  $1$, that is, the supernodes of $A'$ adjacent to $j'$; as $M$ is \emph{induced}, only
  $c'$ qualifies. So the coefficient is $1$ for $A'$ and $0$ for $B'$. If $h<k$ and
  $q=p\geq2$ then $|J|>h+1$ and \Cref{lem:exponentform}(2) forces the two to be equal,
  contradicting $1\neq0$ in $\z_p$; if $q$ is not a multiple of $p$ then
  \Cref{lem:exponentform}(1) forces $q\cdot1=0$ in $\z_p$, i.e. $p\mid q$. Either way the
  $2^{|U'|}\geq2^{\lfloor n/12\rfloor}$ cofactors are pairwise
  $\gen{C^{\leq h}R_q,X}$-inequivalent, hence occupy that many nodes at level $d$.
\end{proof}

\begin{proof}[Proof of \Cref{cor:Xhierarchy}]
  In each case the easy half $L_2\leq_sL_1$ is a size-preserving group inclusion:
  $\gen{C^{\leq h}R_p,X}\subseteq\gen{C^{\leq k}R_q,X}$ whenever $h\leq k$ and $p\mid q$,
  since $C^jR_p=(C^jR_q)^{q/p}$. Both hypotheses are needed:
  $\operatorname{diag}(1,\omega_p)$ is a $\gen{R_q,X}$-LIM only if $p\mid q$, and
  $\gen{C^{\leq h}R_q}$ realizes degree at most $h+1$ and no more. For the hard half,
  \Cref{lem:biasedtower} supplies the tower and \Cref{thm:Xcolumn} the lower bound:
  (1) is its first branch with $p=q$; (2) and (4) its second, which needs $q\nmid p$ in
  the direction shown --- immediate from $p\mid q$, $p<q$ in (2), and hypothesised in
  both directions in (4). For (3) separate with $\ket{\psi^n_{k+1,s,r}}$ for any $s$ of
  which $q$ is not a multiple, say $s=2q$; its tower lies in $C^{\leq k}\Rot$. Finally
  $p=1$ in (2) reads $\gen{C^{\leq k}R_1,X}=\gen X$, and $q\nmid1$ for $q\geq2$.
\end{proof}

Two remarks on the proof itself. First, it uses only two properties of $U'$: that each of
its supernodes has a partner all of whose subnodes lie outside $X_d$, and that no two of
them have a common neighbor. The second is weaker than being an induced matching ---
pairwise grid distance at least $3$ suffices --- and the formalisation proves that weaker
statement from scratch,
so it does not depend
on \Cref{lem:matching}. Second, the step from "pairwise inequivalent cofactors" to "that
many nodes" is not free once $X$ is in the group, because the node a prefix reaches is not
the one the assignment names: the bit flips accumulated along the path permute which child
one descends into. The formalisation carries that iterate explicitly
and derives the bound for an arbitrary
diagram, in its number of \emph{stored} nodes.

\begin{remark}\label{rem:Xhonest}
  The witnesses are biased hypergraph states and not hypergraph states, which is needed to make the separation proof work: $\ket{\phi^G_{p,r}}$ has amplitudes of $m+1$ different
  moduli, and the exact correspondence of \Cref{thm:tower_is_hypergraphstate} needs every
  edge label to have scalar component $1$, which $r\neq1$ violates. Nor is
  \Cref{thm:Xcolumn} really a statement about the bit flip: the bias is chosen precisely
  so that $X$ is provably useless (\Cref{lem:rigidity}), so what the proof shows is that
  the bit-flip column has the same \emph{shape} as the diagonal one, on a family where
  the extra freedom has been switched off. The open question is the unimodular one ---
  whether the separations survive for families all of whose amplitudes have equal
  modulus, hypergraph states in particular. For $p=2,k=0$ they do not, by the stabilizer
  argument of \Cref{sec:results}; for $k\geq1$ or $p\geq3$ we do not know.
\end{remark}

\section{Tractability}\label{sec:tractability}

This section proves the entries of \Cref{tab:tractability}.

\subsection{Queries}\label{sec:queries}

\emph{Measure} and \emph{Sample} both reduce to two recursions down the diagram.
The norm satisfies
$\bra v\ket v=|\lambda_0|^2\bra{v_0}\ket{v_0}+|\lambda_1|^2\bra{v_1}\ket{v_1}$,
with $\lambda_b$ the scalar part of the $b$-edge label, so one pass computes it for
all nodes. The amplitude $\bra x\ket\phi$ is then read off a single walk down the
diagram, accumulating the LIM seen so far. The one point that needs care is which
child the walk enters. If the accumulated LIM at level $j$ is $g$, let
$\tilde x_j\in\{0,1\}$ be the exponent of $X_{(j)}$ in $g$, that is, whether $g$
flips qubit $j$; it is not the queried bit $x_j$. By \Cref{lem:split} the walk
enters the child $e_{x_j\oplus\tilde x_j}$. Sampling follows the same path with a
coin of bias $\bra{e_{\tilde x_j}}\ket{e_{\tilde x_j}}/\bra v\ket v$ at level $j$.
\Cref{app:tract} gives the details.

\emph{Equality} is polytime whenever
a canonical form can be obtained in polytime,
so \Cref{thm:polytime} settles it for $\gen{C^{\leq k}R_q}$- and
$\gen{C^{\leq k}R_q,X}$-LIMDD, for every $q$; the remaining rows are
from~\cite{limdd,hong2025limtdd}.

\emph{InnerProd} and \emph{Fidelity} inherit the hardness proof
of~\cite[Lemma 31]{vinkhuijzen2024a}. That proof needs only that Dicke
states and graph states have polynomial size, which holds already for
$\gen Z$-LIMDD.\footnote{Dicke states are covered by the proof
of~\cite[Lemma 31]{vinkhuijzen2024a} itself, graph states by~\cite{limdd}.}
Both $\gen{C^{\leq k}R_q}$ and $\gen{C^{\leq k}R_q,X}$ contain $Z=R_2$ when
$2\mid q$, for every $k$, which gives the entries marked $\circ^*$. For odd $q$ the
group contains no $Z$ and the complexity is open.

\subsection{Gates}\label{sec:gates-tract}

A gate $U\in\gen G$ is applied to a $\gen G$-LIMDD by multiplying the root edge
label by $U$. This takes constant time and changes no node, and it already covers
a lot of \Yes\ entries in \Cref{tab:tractability}. Concretely, this is the case for gate
$X$, which can be applied to $\gen{X,Z},\gen{R_q,X},\gen{C^{\leq k}R_q,X}$-LIMDD in this way, and for gate $Z$
which can be applied to all groups when $q$ a multiple of $2$, and for gate $T$, 
which can be applied to all groups when $q$ a multiple of $4$. 
For all other cases, we use a method to recursively apply the gate in the $\gen G$-LIMDD. With this method, we modify the
labels on the edges of the LIMDD and increase the number of nodes by a factor that we will see remain constant.

\begin{lemma}[recursive application]\label{lem:recursive}
  Let $U$ be a gate and suppose there are sets $\mathcal V_n,\dots,\mathcal V_1$ of
  gates with $U\in\mathcal V_n$ such that for every node $v$ of index $b$ and every
  $U'\in\mathcal V_b$,
  \begin{equation}\label{eq:recursive}
    U'\ket v=\ket0\otimes L_0'\left(U^{(0)}\ket{v_{j_0}}\right)
            +\ket1\otimes L_1'\left(U^{(1)}\ket{v_{j_1}}\right)
  \end{equation}
  for some $G$-LIMs $L_0',L_1'$, some $U^{(0)},U^{(1)}\in\mathcal V_{b-1}$ and some
  $j_0,j_1\in\{0,1\}$. Then $U$ is applied to a $\gen G$-LIMDD in time linear in the
  size of the diagram, in $\max_b|\mathcal V_b|$ and in the cost of one label
  update, and the diagram grows by at most a factor $\max_b|\mathcal V_b|$.
\end{lemma}
\begin{proof}
  Memoise on the pair (node, element of $\mathcal V_b$). By \eqref{eq:recursive}
  each such pair produces two child pairs, so a node of index $b$ gets at most
  $|\mathcal V_b|$ copies, and every copy is computed once.
\end{proof}

The condition \eqref{eq:recursive} holds for every gate $X,Z,T,R_q,C^k_mR_q$ applied to any 
$\gen G$-LIMDD considered in \Cref{tab:tractability}, though we only need it in the cases
where the gate can not simply be applied to the root edge. Crucially, for these non-trivial cases,
$\max_b|\mathcal V_b|$ is bounded by $3$, so the recursive application of the gate in the LIMDD
increases its number of nodes by a constant factor only.

\begin{lemma}[gate transfer]\label{lem:transfer}
  Let $\gen G$ be any of the five classes of \Cref{tab:tractability}. Then
  \eqref{eq:recursive} holds for $U=X_{(a)}$ with $\mathcal V_b=\{X_{(a)}\}$ for
  $b\geq a$ and $\mathcal V_b=\{I\}$ below, and for
  $U=\left(C^{k'}R_m\right)^d$ placed on qubits $a_1<\dots<a_{k'+1}$, with $k'\geq0$
  and $m,d$ arbitrary, with
  \begin{equation*}
    \mathcal V_b=\begin{cases}
      \bigl\{(C^{k'}R_m)^{\pm d}_{(a_1,\dots,a_{k'+1})}\bigr\}
        & b\geq a_{k'+1},\\
      \bigl\{(C^{k''}R_m)^{\pm d}_{(a_1,\dots,a_{k''+1})}\bigr\}\cup\{I\}
        & a_{k''}\leq b<a_{k''+1},\ k''\leq k',\\
      \{I\} & b<a_1 .
    \end{cases}
  \end{equation*}
  The exponent $-d$ occurs only when $X\in\gen G$, so
  $\max_b|\mathcal V_b|\leq3$, and $\leq2$ for a diagonal $\gen G$. Hence every
  \Yes\ entry of the gates in \Cref{tab:tractability} is applicable in polytime, with the diagram
  growing by at most a factor $3$. The case $k'=0$ is $(R_m)^d$, which covers $Z$
  and $T$ for $m=2$ and $m=4$ respectively.
\end{lemma}
\begin{proof}[Proof idea]
  Two facts drive every case. Below the lowest control, at $b<a_1$, the gate acts as
  the identity. At a control level $b=a_{k''+1}$ the gate loses one control and
  descends only into the $1$-branch, which is the case $\ell\mapsto\ell-1$ of
  \eqref{eq:commute}. Off a control level the gate must be pushed past the edge
  label: for a diagonal $\gen G$ the two commute outright and $\mathcal V_b$ needs
  only two elements, while for a $\gen G$ containing $X$ the label may flip a
  control, and \eqref{eq:commute} then returns the same gate with $d$ replaced by
  $-d$. That inversion is the only extra element, so three suffice. The six
  instantiations, one per row and gate family, are in \Cref{app:tract}.
\end{proof}

The three \No\ columns all follow from one lower bound. For an integer $q$ put
\begin{equation*}
  \ket{rot^n}=\bigotimes_{j=1}^{n}\left(\ket0+e^{\pi i2^{-j}/q}\ket1\right),
  \qquad
  \ket{sum^n}=\ket0\otimes\ket+^{\otimes n}\otimes\ket0
             +\ket1\otimes\ket{rot^n}\otimes\ket0 .
\end{equation*}

\begin{theorem}\label{thm:swap}
  Fix the variable order. For every $k$ and $q$, the state $\ket{sum^n}$ has a
  linear-size $\gen I$-LIMDD, while
  $\textit{SWAP}_{1,n+2}\ket{sum^n}$ needs a $\gen{C^{\leq k}R_q,X}$-LIMDD with at
  least $2^{n-1}$ nodes. Consequently $\textit{SWAP}$, $CX$ and $H$ are not applicable in
  polytime to any class of \Cref{tab:tractability}.
\end{theorem}
\begin{proof}[Proof idea]
  A linear-size $\gen I$-LIMDD for $\ket{rot^n}$ and for $\ket+^{\otimes n}$ is
  given by~\cite[Lemma 22]{vinkhuijzen2024a}, and enlarging $\gen G$ never
  enlarges a diagram, so $\ket{sum^n}$ is linear in every class of LIMDD in
  \Cref{tab:tractability}. After the swap, the nodes of the second-to-last layer represent
  the two-qubit functions
  \begin{equation}\label{eq:fa}
    f_a(x_{n+1},x_{n+2})
      =(1-x_{n+2})+x_{n+2}\cdot e^{\pi ix_{n+1}2^{-n}/q}
        \cdot e^{\pi i\sum_{j=1}^{n-1}a_{j+1}2^{-j}/q},
  \end{equation}
  one for every $a\in\{0,1\}^n$ with $a_1=0$: fixing $x_1=0$ and
  $x_2\dots x_n=a_2\dots a_n$ leaves the $x_{n+2}=0$ branch constant $1$ and turns
  the $x_{n+2}=1$ branch into the product of $\ket{rot^n}$'s phases, of which the
  factors $j\leq n-1$ are constant and the factor $j=n$ depends on $x_{n+1}$. Two
  distinct such $f_a,f_b$ are not $\gen{C^{\leq k}R_q,X}$-equivalent. The node has
  index $2$, so every LIM acting on it has at most one control whatever $k$ is, and
  the four possible bit-flip patterns are ruled out one by one by a $2$-adic
  valuation argument on the required exponent. \Cref{app:tract} gives the four
  computations. The succinctness relations from \Cref{sec:succinctness} show that 
  $\gen{C^{\leq k}R_q,X}$ is most compact among all classes in \Cref{tab:tractability}, 
  so the lower bound applies to all of them.
  Finally $\textit{SWAP}_{(a,b)}=CX_{(a,b)}CX_{(b,a)}CX_{(a,b)}$ makes $CX$
  intractable, and $CX_{(a,b)}=H_{(b)}CZ_{(a,b)}H_{(b)}$ with $CZ$ tractable by
  \Cref{lem:transfer} makes $H$ intractable.
\end{proof}

\section{Canonicity}\label{sec:canonicity}

\subsection{The exponent-vector encoding and the layer split}\label{sec:machinery}

The canonicity proofs for $\gen{C^{\leq k}R_q}$-LIMDD and $\gen{C^{\leq k}R_q,X}$-LIMDD 
need a more detailed study of the (algebraic) structure 
of the groups $\gen{C^{\leq k}R_q}$ and $\gen{C^{\leq k}R_q,X}$, which we will describe in this section.
\Cref{app:group} extends description of the group structure for $\gen{C^{\leq k}R_q,X}$, 
but the following is sufficient to understand the canonicity proofs in \Cref{sec:canonicity}.

By \Cref{lem:commute} every element of $\gen{C^{\leq k}R_q,X}$ can be written with
all bit flips $X$ on the right. Together with the commutativity of diagonal matrices
this gives a standard form: every $\gen{C^{\leq k}R_q,X}$-LIM equals
\begin{equation}\label{eq:standard-form}
  \lambda\cdot
  \prod_{\ell=1}^{k+1}\ \prod_{a\in\binom{[n]}{\ell}}
    \left[(C^{\ell-1}R_q)_{(a)}\right]^{d_a}
  \prod_{j=1}^{n}\left[X_{(j)}\right]^{x_j},
  \qquad d_a\in\z_q,\ x_j\in\z_2,\ \lambda\in\co ,
\end{equation}
and a $\gen{C^{\leq k}R_q}$-LIM is the special case $x_1=\dots=x_n=0$. The
exponents in \eqref{eq:standard-form} are unique, and multiplication of $\gen{C^{\leq k}R_q}$-LIMs is commutative as all such LIMs are diagonal matrices. Writing
\begin{equation}\label{eq:M}
  M:=\sum_{\ell=1}^{k+1}\binom{n}{\ell}\in O(n^{k+1})
\end{equation}
for the number of nonempty index sets involved, every LIM in $\gen{C^{\leq k}R_q}$ can be represented by a length $M$ vector over $\z_q$. Hence, we can formally write the isomorphism
\begin{equation}\label{eq:iso}
  \gen{C^{\leq k}R_q}\;\cong\;\z_q^{M} .
\end{equation}

We order the $M$ components by \emph{decreasing} arity,
\begin{equation}\label{eq:order}
  \left(d_a\right)_{|a|=k+1},\ \left(d_a\right)_{|a|=k},\ \dots,\ \left(d_a\right)_{|a|=1},
\end{equation}
and within each component we order each entry in lexicographic order of $a$, and we compare LIMs by the lexicographic
order on $(\lambda,(x_j)_j,(d_a)_a)$ induced by \eqref{eq:order}; we write
$g\leq_{\mathrm{lex}}h$ for it. This is an order on \emph{LIMs}, and is unrelated to the
total order $\preccurlyeq$ on \emph{nodes} of \Cref{sec:rules}. Minimizing in this
order therefore minimizes the components of top arity $|a|=k+1$ before the lower arity components, which is what
\Cref{sec:algo} needs.

We will repeatedly need to know how a single group element acts on the two
cofactors of a node. Split $g\in\gen{C^{\leq k}R_q,X}$ as
\begin{equation}\label{eq:split}
  g=g'\,g''\,X_{(1)}^{x_1}\,g''' ,
\end{equation}
where $g'$ collects the diagonal factors $C^{\ell-1}R_q$ of \eqref{eq:standard-form} whose index set
contains qubit $1$, $g''$ collects those whose index set does not contain qubit $1$, $x_1$ is a Boolean variable, and $g'''=\prod_{j=2}^{n}[X_{(j)}]^{x_j}$ are the $X$ factors on qubit $2$ to $n$. For convenience, we add the scalar $\lambda$ to $g''$. Both $g''$ and
$g'''$ act as the identity on qubit $1$. The factor $g'$ acts as the identity on
$\ket0\otimes\co^{2^{n-1}}$, and on $\ket1\otimes\co^{2^{n-1}}$ it acts as a single
element $\tilde{g}$ of $\gen{C^{\leq k-1}R_q}$ on qubits $2,\dots,n$, namely
\begin{equation}\label{eq:gtilde}
  g'\left(\ket1\otimes I^{\otimes(n-1)}\right)=\ket1\otimes\tilde g,
  \qquad
  \tilde g=\omega_q^{d_{\{1\}}}\!\!\prod_{\ell=2}^{k+1}\prod_{a\in\binom{[2,n]}{\ell-1}}
  \left[(C^{\ell-2}R_q)_{(a)}\right]^{d_{\{1\}\cup a}} .
\end{equation}
Conversely $g'$ is recovered from $\tilde g$ by adding a control on qubit $1$ to
every factor. We write $\ctl(\tilde g):=g'$ for this operation. It is the
$\z_q$-linear map on exponent vectors determined by
$\pi_{\{1\}\cup a}(\ctl(\tilde g))=\pi_a(\tilde g)$, and $\pi_a(\ctl(\tilde g))=0$
if $1\notin a$, and it is computable with $O(M)$ operations in $\z_q$. The
following lemma is the only calculation the rest of the paper needs.

\begin{lemma}[layer split]\label{lem:split}
  Let $g=g'\,g''\,X_{(1)}^{x_1}\,g'''\in\gen{C^{\leq k}R_q,X}$ be decomposed as in \eqref{eq:split}, and let
  $\ket u,\ket w$ be $(n-1)$-qubit vectors. Then
  \begin{equation*}
    g\left(\ket0\otimes\ket u+\ket1\otimes\ket w\right)=
    \begin{cases}
      \ket0\otimes g''g'''\ket u+\ket1\otimes \tilde g\,g''g'''\ket w
        &\text{if }x_1=0,\\[2pt]
      \ket0\otimes g''g'''\ket w+\ket1\otimes \tilde g\,g''g'''\ket u
        &\text{if }x_1=1.
    \end{cases}
  \end{equation*}
  Moreover, as $g$ ranges over $\gen{C^{\leq k}R_q,X}$, the factor $\tilde g$ ranges
  over all of $\gen{C^{\leq k-1}R_q}$ on qubits $2,\dots,n$, the factor $g''g'''$
  ranges over all of $\gen{C^{\leq k}R_q,X}$ on qubits $2,\dots,n$, and the two
  vary \emph{freely and independently} of each other. For $g\in\gen{C^{\leq k}R_q}$
  the same holds with $x_1=0$ and $g'''=I$.
\end{lemma}

Note that the factor $\tilde g$ acting on the high cofactor is unconstrained by
the factor acting on the low cofactor. \Cref{sec:rules} shows that this free factor
is exactly what a correct high-edge minimization must range over.

\Cref{fig:makeedge} recalls the boundary this split sits on: $\mathsf{makeEdge}$
turns two child edges into one reduced node plus a root label, and every rule of
\Cref{sec:rules} acts inside that step.

\begin{figure}[t]
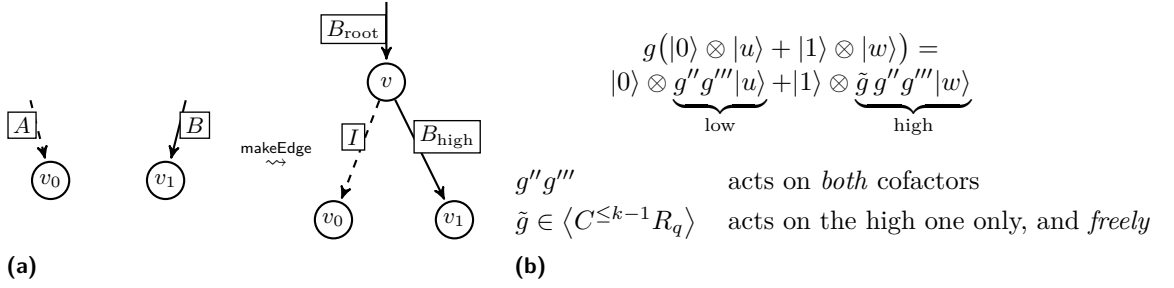

\centering
\begin{subfigure}[b]{0.44\textwidth}
\centering
\tikz[->,>=stealth',shorten >=1pt,auto,node distance=1.5cm,font=\footnotesize,
      thick, state/.style={circle,draw,inner sep=0pt,minimum size=14pt}]{
  \node[state] (v0) {$v_0$};
  \node[state, right = 1.0cm of v0] (v1) {$v_1$};
  \node[above = .8cm of v0, xshift=-.3cm] (a0) {};
  \node[above = .8cm of v1, xshift=.3cm] (a1) {};
  \path[]
    (a0) edge[e0] node[pos=.4,lbl,solid,left]  {$A$} (v0)
    (a1) edge[e1] node[pos=.4,lbl,right] {$B$} (v1);

  \node[state, right = 2.4cm of v1, yshift=1.3cm] (v) {$v$};
  \node[above = .8cm of v] (r) {};
  \node[state, below = 1.3cm of v, xshift=-.7cm] (w0) {$v_0$};
  \node[state, below = 1.3cm of v, xshift=.9cm]  (w1) {$v_1$};
  \path[]
    (r) edge node[pos=.45,lbl,left] {$\rootlim$} (v)
    (v) edge[e0] node[pos=.35,lbl,solid,left]  {$\id$}      (w0)
    (v) edge[e1] node[pos=.35,lbl,right] {$\highlim$} (w1);
  \node at ($(v1)!0.5!(v)+(0,-0.3)$) {$\overset{\mathsf{makeEdge}}{\rightsquigarrow}$};
}
\caption{}\label{fig:makeedge-a}
\end{subfigure}
\hfill
\begin{subfigure}[b]{0.52\textwidth}
\centering
$g\bigl(\ket0\otimes\ket u+\ket1\otimes\ket w\bigr)
 =\ket0\otimes\underbrace{g''g'''\ket u}_{\text{low}}
  +\ket1\otimes\underbrace{\tilde g\,g''g'''\ket w}_{\text{high}}$
\\[2.2ex]
\begin{tabular}{@{}ll@{}}
  $g''g'''$ & acts on \emph{both} cofactors\\[2pt]
  $\tilde g\in\gen{C^{\leq k-1}R_q}$ & acts on the high one only, and \emph{freely}
\end{tabular}
\caption{}\label{fig:makeedge-b}
\end{subfigure}
\caption{\subref{fig:makeedge-a} $\mathsf{makeEdge}$ takes the two child edges
$(A,v_0)$ and $(B,v_1)$ and returns a single edge $(\rootlim,v)$ whose target
satisfies R1--R5: R3 puts $\id$ on the $0$-edge, R4 replaces the $1$-edge label
by the $\leq_{\mathrm{lex}}$-least $\highlim$ of its class, and whatever the two
rules divide out is collected in the root label $\rootlim$.
\subref{fig:makeedge-b} the layer split of \Cref{lem:split} is what keeps this
local: a group element acts on the two cofactors through one \emph{shared}
factor $g''g'''$ and one \emph{free} factor $\tilde g$ on the high cofactor.
That free factor is exactly what the first guess of \Cref{sec:rules} misses.}
\label{fig:makeedge}
\end{figure}

\subsection{Reduction rules}\label{sec:rules}

Fix a total order $\preccurlyeq$ on nodes, arbitrary but the same for all diagrams
under consideration.\footnote{A per-diagram topological order does not suffice. The
uniqueness argument of \Cref{sec:canonform} compares two different diagrams and
applies antisymmetry of $\preccurlyeq$ across them.}

Recall from \Cref{sec:results} the five rules and the notation: $v$ is a node
whose $0$-edge points to $v_0$ with label $L_0$ and whose $1$-edge points to
$v_1$ with label $L_1$.
The first three rules are immediate rewritings. For R1,
$\ket b\otimes L_b\ket{v_b}+\ket{1-b}\otimes L_{1-b}\ket{v_{1-b}}
=\ket b\otimes L_b\ket{v_{1-b}}+\ket{1-b}\otimes L_{1-b}\ket{v_{1-b}}$ when $L_b=0$.
For R2, swap the two children and multiply every incoming edge label by $X_{(1)}$,
using
$\ket0\otimes L_0\ket{v_0}+\ket1\otimes L_1\ket{v_1}=(X\otimes I)(\ket0\otimes L_1\ket{v_1}+\ket1\otimes L_0\ket{v_0})$.
For R3 with $L_0\neq0$, use the factoring of \Cref{sec:prelim}. Each of the three
updates the labels of the incoming edges, and no other node is touched.
\Cref{fig:reduced1} showed the resulting node for two of the cases.

Rule R4 is the only rule that is not trivially clear it can be applied in polytime. Getting its domain of minimization right is the point of
this subsection. 
A first guess is the set
$Stab(v_0)^{-1}L_1Stab(v_1)$, similar to the Pauli-LIMDD canonicity algorithm: replacing $L_1$ by $g_0^{-1}L_1g_1$ with
$g_i\in Stab(v_i)$ leaves $\ket v$ fixed up to the LIM $I\otimes g_0$. That set is
too small. By \Cref{lem:split}, a general $g\in\gen G$ acts on the two cofactors as
$g''g'''$ and as $\tilde gg''g'''$, and $\tilde g$ varies over
$\gen{C^{\leq k-1}R_q}$ freely, independently of $g''g'''$. Repeating the substitution
of the previous sentence with such a general $g$ in place of $I\otimes g_0$ therefore
produces one further factor.

\propfour*
\begin{proof}
  Apply \Cref{lem:split} to a general $g\in\gen G$ relating the two nodes $\ket v=\ket0\otimes\ket{v_0}+\ket1\otimes L_1\ket{v_1}$ and
  $\ket{v'}=\ket0\otimes\ket{v_0}+\ket1\otimes L_1'\ket{v_1}$, i.e.\ with $g\ket v=\ket{v'}$. 

  This gives for $x_1=0$:
  \begin{equation*}
    g\left(\ket0\otimes\ket{v_0}+\ket1\otimes L_1\ket{v_1}\right)
      =\ket0\otimes g''g'''\ket{v_0}+\ket1\otimes \tilde gg''g'''\,L_1\ket{v_1}
      =\ket0\otimes\ket{v_0}+\ket1\otimes L_1'\ket{v_1} .
  \end{equation*}
  So the low cofactor is sent to itself by $g''g'''$, i.e. $g''g'''\in Stab(v_0)$.
  The edge label $L_1'$ is thus $\tilde g(g''g''')\,L_1 h$ for some $h\in Stab(v_1)$.
  This gives the first coset of \eqref{eq:true_label_class}, because $\tilde g$ ranges 
  over all of $\gen{C^{\leq k-1}R_q}$ independently of $g''g'''$.

  For $x_1=1$ we get:
  \begin{equation*}
    g\left(\ket0\otimes\ket{v_0}+\ket1\otimes L_1\ket{v_1}\right)
      =\ket0\otimes g''g'''\,L_1\ket{v_1}+\ket1\otimes \tilde gg''g'''\ket{v_0}
      =\ket0\otimes\ket{v_0}+\ket1\otimes L_1'\ket{v_1} .
  \end{equation*}
  So $\ket{v_0}$ and $\ket{v_1}$ are exchanged, and hence $\gen G$-inequivalent, which means that $v_0=v_1$ by R5.
  Again, the low cofactor is sent to the high cofactor by $g''g'''\,L_1$, i.e. $g''g'''\,L_1\in Stab(v_0)=Stab(v_1)$. 
  The edge label $L_1'$ is thus $\tilde g(g''g''') h$ for some $h\in Stab(v_0)=Stab(v_1)$ and $\tilde g\in\gen{C^{\leq k-1}R_q}$,
  which gives the second coset of \Cref{prop:r4class} for the case that $v_0=v_1$.
\end{proof}

\subsection{The canonical form}\label{sec:canonform}

Two remarks on the statement of \Cref{thm:canonicity} (canonicity), restated below.
First, the all-zero vector must be excluded: by R3 the two edge labels of a node cannot both be $0$, so no node
satisfying R1--R5 represents it. Second, the map relating $\ket\phi$ to the root node is a $G$-LIM
$\ell=\lambda\cdot g$ and not merely a group element $g$. Rule R3 forces every $0$-edge
label to be $I^{\otimes n}$, so the state $\ket v$ of the root node always has first
nonzero amplitude $1$, whereas by definition
every element of $\gen G$ has all entries of unit modulus for $G\in\{\gen{C^{\leq k}R_q},\gen{C^{\leq k}R_q,X}\}$. Already for $n=1$ the vector
$2\ket0+4\ket1$ is not of the form $g\ket v$ with $g\in\gen G$.

\thmcanonicity*
\begin{proof}[Proof idea]
  We argue for $\gen{C^{\leq k}R_q,X}$; the diagonal case $\gen{C^{\leq k}R_q}$ is the
  same argument with the reordering of the children and every case $x_1=1$ deleted,
  since a diagonal $g$ has $x_1=0$ in \eqref{eq:split} and R2 is then vacuous.
  Induction on $n$. For existence, the two cofactors have canonical nodes
  $v_0,v_1$ by induction; put the smaller first, using $X$ on the incoming labels if
  needed, and let R4 fix the high label. For uniqueness, let $g_v\ket v=g_w\ket w$
  and apply \Cref{lem:split} to $g_w^{-1}g_v$. If $x_1=0$ the children match
  pairwise and induction identifies them. If $x_1=1$ the children are exchanged, so
  R2 gives $v_0\preccurlyeq v_1=w_0\preccurlyeq w_1=v_0$ and antisymmetry of
  $\preccurlyeq$ collapses all four to one node. In both cases the two high labels
  lie in the same set \eqref{eq:true_label_class}, the second case using the
  inverse-twisted coset, so R4 selects the same one and R5 merges. \Cref{app:canon}
  gives the details.
\end{proof}

The inverse-twisted coset is used exactly once, in the case $x_1=1$, and it is
precisely why the domain of R4 must include it.

\subsection{Computing the canonical form}\label{sec:algo}

\thmpolytime*

Three remarks on the statement. The size of the LIMDD is the number of \emph{stored} nodes of
the shared diagram. A bound in the size of the tree unfolding would be exponentially
weaker, and the formalisation exhibits a diagram whose unfolding is exponentially
larger than the diagram itself. Next, every \emph{root} is assumed to denote a
nonzero state. This is not a restriction of the method but a characterization of the
diagrams that have a reduced form at all, and it constrains only the top layer; a
diagram may store vanishing nodes beneath it. Finally, the bound covers R4 and R5 at
every stored node, and R1 and R3 add two bit operations per node: on a stored diagram
their dispatch is decided by the edge-label tags alone. Rule R2 is vacuous without
$X$, and with $X$ its comparison of the two children is not counted: the labels
carry complex scalars, and the model therefore has no computable order on nodes to
count steps against. A scalar
comparison counts as one operation, as elsewhere in this paper. Recall from
\Cref{sec:results} that ``polynomial in the size of the LIMDD'' means polynomial
for each fixed $k$: by \eqref{eq:M} a single label already has $M\in O(n^{k+1})$
digits, so the size of the diagram and the computational cost are both exponential in $k$.

\begin{proof}
  We give the algorithm for $\gen{C^{\leq k}R_q}$. The algorithm for
  $\gen{C^{\leq k}R_2,X}$ is in \Cref{app:algoX}.

  Rules R1--R3 are the rewritings of \Cref{sec:rules}. Each performs $O(1)$ label
  products, and a label product is $O(M)$ additions in $\z_q$ by \eqref{eq:iso}.
  Rule R5 is a lookup in a hash table on the tuple of children and labels.

  For R4 we use the identification $\gen{C^{\leq k}R_q}\cong\z_q^{M}$ of
  \eqref{eq:iso}, with components ordered by decreasing arity as in
  \eqref{eq:order}. A subgroup of $\z_q^{M}$ is a submodule, and is represented by a
  matrix in Howell normal form, which is a canonical generating set for the row
  span~\cite{howell1986spans}. For prime $q$, the Howell normal form coincides with the reduced row echelon form.
  The sum, the intersection, and the preimage under a $\z_q$-linear map of such
  submodules, and the
  $\leq_{\mathrm{lex}}$-least element of a coset $u+S$ of a submodule $S\leq\z_q^{M}$,
  are computable with $O(M^{3})$
  operations in $\z_q$ by Howell reduction~\cite{storjohann1998modn,storjohann2000canonical}.\footnote{There are actually ways to do this with $O(M^{\omega})$, where $\omega$ is the matrix multiplication constant. In practice, the number of operations scales as $O(M^3)$ for small $M$.}

  \emph{Stabilizers.} Let $g\in Stab(v)$ and split it as $g=\ctl(\tilde g)g''$
  following \eqref{eq:split}. \Cref{lem:split} gives
  $\ket v=\ket0\otimes g''\ket{v_0}+\ket1\otimes L_1\tilde gg''\ket{v_1}$, using
  that $\gen{C^{\leq k}R_q}$ is abelian to move $L_1$ out. Comparing with
  $\ket v=\ket0\otimes\ket{v_0}+\ket1\otimes L_1\ket{v_1}$ yields $g''\in Stab(v_0)$
  and $\tilde gg''\in Stab(v_1)$.

  Writing $P(g)=g''$ and
  $T(g)=\tilde gg''$ for the two $\z_q$-linear comparison maps of \eqref{eq:split},
  the two conditions above read $P(g)\in Stab(v_0)$ and $T(g)\in Stab(v_1)$, i.e.
  \begin{equation}\label{eq:stabcomap}
    Stab(v)=P^{-1}\bigl(Stab(v_0)\bigr)\cap T^{-1}\bigl(Stab(v_1)\bigr)
  \end{equation}
  which is two preimages and one intersection
  of submodules of $\z_q^{M}$, hence computable in time polynomial in $M$ by the
  primitives above.

  \emph{Minimization.} Since $\gen{C^{\leq k}R_q}$ is abelian and $Stab(v_0),Stab(v_1)$ and $\gen{C^{\leq k-1}R_q}$ are 
  groups, the right-hand side of \eqref{eq:true_label_class} is the coset $HL_1$ of
  \begin{equation}\label{eq:H}
    H:=\gen{\gen{C^{\leq k-1}R_q},Stab(v_0),Stab(v_1)} .
  \end{equation}
  We can compute $H$ as a sum of submodules in polytime using the Howell normal form.
  Then, computing the $\leq_{\mathrm{lex}}$-least element
  of $HL_1$ can also be done in time polynomial in $M$.
\end{proof}

  \emph{Standard form of canonical 1-edge label} Observe $\gen{C^{\leq k-1}R_q}\leq H$. Its elements
  are exactly the LIMs supported on the components $a$ with $|a|\leq k$, and they
  act on the exponent vector by translation. Those components therefore form a free
  orbit direction: every value is attained, there is nothing to minimize, and all
  of them are set to $0$. The canonical high edge label therefore consists of the
  scalar $\lambda$ and the $\binom{n}{k+1}$ components of arity $k+1$, and nothing
  else --- and, in the bit-flip case of the appendix, the $\gen X$-part as
  well. In particular, on a node of $n\leq k+1$ qubits nothing of arity $k+1$
  survives either, so the high edge label is always $\lambda\cdot I^{\otimes n}$.
  The global phase $d_\emptyset$ is free and is set to $0$.

\section{The quantum Fourier transform is a tower}\label{sec:qft}

We prove \Cref{thm:qft_is_tower}, in the matrix-as-state convention fixed in
\Cref{sec:results}: the quantum Fourier transform matrix has a linear-size diagram over
$\gen{R_N}$ alone, with $N=2^n$.

\thmqfttower*
\begin{proof}
  The Fourier matrix is $(\omega_N^{jk})_{j,k}$ with $\omega_N=e^{2\pi i/N}$.
  Writing $\hat x=\sum_{d=0}^{n-1}2^dx_d$ and likewise $\hat x'$ gives
  $\omega_N^{\hat x\hat x'}=\omega_N^{\sum_{d,d'=0}^{n-1}2^{d+d'}x_dx'_{d'}}$. Put
  \begin{equation*}
    \ket{\phi_\ell}=\sum_{x,x'}\omega_N^{\sum_{d,d'=0}^{\ell}2^{d+d'}x_dx'_{d'}}\ket{x,x'} .
  \end{equation*}
  Splitting the exponent on the two variables of level $\ell$ separates the terms
  with $d=\ell$ or $d'=\ell$ from the rest, which gives
  \begin{equation*}
    \ket{\phi_\ell}=\ket{00}\otimes\ket{\phi_{\ell-1}}
      +\ket{01}\otimes R'\ket{\phi_{\ell-1}}
      +\ket{10}\otimes R\ket{\phi_{\ell-1}}
      +\ket{11}\otimes \omega_N^{2^{2\ell}}RR'\ket{\phi_{\ell-1}},
  \end{equation*}
  with
  \begin{align*}
    R &= I\otimes R_N^{2^{2\ell-1}}\otimes I\otimes R_N^{2^{2\ell-2}}\otimes\cdots\otimes I\otimes R_N^{2^{\ell}},\\
    R'&= R_N^{2^{2\ell-1}}\otimes I\otimes R_N^{2^{2\ell-2}}\otimes\cdots\otimes I\otimes R_N^{2^{\ell}}\otimes I .
  \end{align*}
  All four cofactors are $\ket{\phi_{\ell-1}}$ up to a $\gen{R_N}$-LIM, so both
  edges of the node for $x_\ell$ point to the node for $x'_\ell$, and both edges of
  the latter point to the node for $x_{\ell-1}$. This holds for every $\ell$, so the
  diagram is a tower, with $1$-edge labels $RR'\omega_N^{2^{2\ell}}$ and $R'$ at the
  two levels of stage $\ell$. The argument never uses the position of level $\ell$
  relative to the others, only that $x_\ell$ and $x'_\ell$ are adjacent.
\end{proof}

\begin{figure}[t]
  \centering
  \begin{tikzpicture}[
    scale=0.3,
    every path/.style={>=latex},
    every node/.style={},
    inner sep=1pt,
    minimum size=0.3cm,
    line width=1pt,
    node distance=.8cm,
    thick,
    font=\scriptsize
    ]
    \node[draw,circle] (a1) {$x_{n-1}$};
    \node[draw,circle, below = .4cm of a1 ] (a2) {$x_{n-1}'$};
    \node[draw,circle, below = .4cm of a2 ] (a3) {$x_{n-2}$};
    \node[draw,circle, below = .4cm of a3      ] (a4) {$x_{n-2}'$};
    \node[draw,circle, below = .4cm of a4      ] (a5) {$x_{n-3}$};
    \node[draw,circle, below = .6cm of a5      ] (a6) {$x_{1}$};
    \node[draw,circle, below = .4cm of a6      ] (a7) {$x_{1}'$};
    \node[draw,circle, below = .4cm of a7      ] (a8) {$x_{0}$};
    \node[draw,circle, below = .4cm of a8      ] (a9) {$x_{0}'$};
    \node[leaf, below = .4cm of a9      ] (a10) {$1$};

    \draw[<-] (a1) --++(90:4cm) node[right,pos=.7] {};

    \draw[e0=20] (a1) edge  node[left] {} (a2);
    \draw[e1=20] (a1) edge  node[right,pos=.3] {~$R_N^{2^{2(n-1)}}\otimes I\otimes R_N^{2^{2(n-1)-1}}\otimes I\otimes\dots\otimes I\otimes R_N^{2^{n-1}}$} (a2);
    \draw[e0=20] (a2) edge  node[left] {} (a3);
    \draw[e1=20] (a2) edge  node[right,pos=.3] {~~$R_N^{2^{2(n-1)-1}}\otimes I\otimes R_N^{2^{2(n-1)-2}}\otimes\dots\otimes I\otimes R_N^{2^{n-1}}\otimes I$} (a3);
    \draw[e0=20] (a3) edge  node[left] {} (a4);
    \draw[e1=20] (a3) edge  node[right,pos=.3] {~$R_N^{2^{2(n-2)}}\otimes I\otimes R_N^{2^{2(n-2)-1}}\otimes\dots\otimes I\otimes R_N^{2^{n-2}}$} (a4);
    \draw[e0=20] (a4) edge  node[left] {} (a5);
    \draw[e1=20] (a4) edge  node[right,pos=.3] {~$R_N^{2^{2(n-2)-1}}\otimes I\otimes R_N^{2^{2(n-2)-2}}\otimes \dots\otimes I\otimes R_N^{2^{n-2}}\otimes I$} (a5);
    \draw[dotted] (a5) -- (a6);
    \draw[e0=20] (a6) edge  node[left] {} (a7);
    \draw[e1=20] (a6) edge  node[right,pos=.3] {~$R_N^4\otimes I\otimes R_N^2$} (a7);
    \draw[e0=20] (a7) edge  node[left] {} (a8);
    \draw[e1=20] (a7) edge  node[right,pos=.3] {~$R_N^2\otimes I$} (a8);
    \draw[e0=20] (a8) edge  node[left] {} (a9);
    \draw[e1=20] (a8) edge  node[right,pos=.3] {~$R_N$} (a9);
    \draw[e0=20] (a9) edge  node[left] {} (a10);
    \draw[e1=20] (a9) edge  node[right,pos=.3] {~$1$} (a10);

\end{tikzpicture}
  \caption{The quantum Fourier transform on $n$ qubits as a
  Tower-$\gen{R_N}$-LIMDD, $N=2^n$. The variable order is reversed for readability.}
  \label{fig:QFT-Tower-LIMDD}
\end{figure}
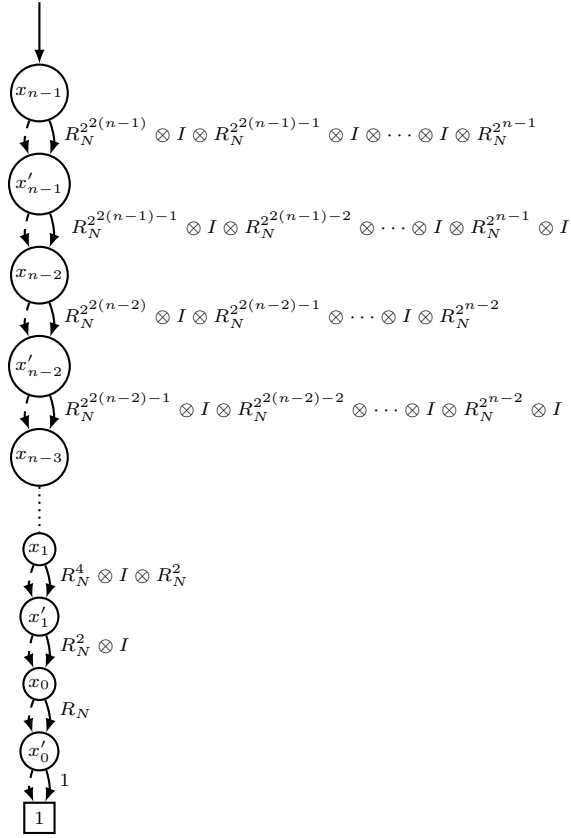

Compact decision diagrams for the Fourier transform are reported experimentally
in~\cite{hong2025limtdd}; \Cref{thm:qft_is_tower} proves the corresponding
statement for Tower-$\gen{R_N}$-LIMDDs.

\paragraph*{Use of AI tools.}
The authors used substantial
assistance from an AI coding assistant (Anthropic's Claude) to formalize and verify some of their results:
it produced candidate Lean proof scripts, and it
was also directed to act as an adversarial referee of the manuscript, searching
for bugs in proofs. That search is what uncovered the
error corrected in \Cref{prop:r4class}: an earlier version of rule R4 minimized
over a domain that omits the free factor $\gen{C^{\leq k-1}R_q}$ produced by the
layer split, and the assistant produced the machine-checked counterexample at
$q=2$, $k=1$, $n=2$. The correctness of every
formalized statement rests on the Lean kernel and not on the assistant: the
development contains no \verb|sorry| and no \verb|native_decide|, and each
headline result was audited with \verb|#print axioms|. The same assistant was
used for editing and restructuring the manuscript, including the present
arrangement of the material into a self-contained introduction. All definitions,
theorem statements and proofs were written by the authors, who take
full responsibility for the content.

\bibliography{bibliography,lits}

\appendix
\section{Omitted proofs of Sections~\ref{sec:prelim} and~\ref{sec:machinery}}\label{app:prelim}

\begin{proof}[Proof of \Cref{lem:commute}]
  Both sides of \eqref{eq:commute} are products of a permutation matrix and a
  diagonal matrix, so it suffices to check them on every basis state
  $\ket{x_1,\dots,x_n}$. Since $(C^{\ell-1}R_q)_{(a)}$ has order $q$, the identity
  is equivalent to
  \begin{equation*}
    (C^{\ell-1}R_q)_{(a)}\,X_{(a_j)}\,(C^{\ell-1}R_q)_{(a)}
      =(C^{\ell-2}R_q)_{(a\setminus\{a_j\})}\,X_{(a_j)} .
  \end{equation*}
  Note that by assumption $a_j\in a$. 
  We observe that $X_{(a_j)}(C^{\ell-2}R_q)_{(a\setminus\{a_j\})}=(C^{\ell-2}R_q)_{(a\setminus\{a_j\})}X_{(a_j)}$, as the gates act on different qubits.
  Then we rewrite
\begin{align*}
&(C^{\ell-1}R_q)_{(a)}X_{(a_j)}(C^{\ell-1}R_q)_{(a)}\ket{x_1,\dots,x_n}\\
&=\begin{cases}
    (C^{\ell-1}R_q)_{(a)}X_{(a_j)}\ket{x_1,\dots,x_n} &\text{ if }x_{a_j} = 0\\
    (C^{\ell-1}R_q)_{(a)}X_{(a_j)}(C^{\ell-2}R_q)_{(a\setminus\{a_j\})}\ket{x_1,\dots,x_n} &\text{ if }x_{a_j} = 1\\
\end{cases}\\
&=\begin{cases}
    (C^{\ell-1}R_q)_{(a)}\ket{x_1,\dots,1-x_{a_j},\dots,x_n} &\text{ if }x_{a_j} = 0\\
    (C^{\ell-1}R_q)_{(a)}(C^{\ell-2}R_q)_{(a\setminus\{a_j\})}\ket{x_1,\dots,1-x_{a_j},\dots,x_n} &\text{ if }x_{a_j} = 1\\
    \qquad\text{using }X_{(a_j)}(C^{\ell-2}R_q)_{(a\setminus\{a_j\})}=(C^{\ell-2}R_q)_{(a\setminus\{a_j\})}X_{(a_j)}\\
\end{cases}\\
&=\begin{cases}
    (C^{\ell-2}R_q)_{(a\setminus\{a_j\})}\ket{x_1,\dots,1-x_{a_j},\dots,x_n} &\text{ if }x_{a_j} = 0\\
    (C^{\ell-2}R_q)_{(a\setminus\{a_j\})}\ket{x_1,\dots,1-x_{a_j},\dots,x_n} &\text{ if }x_{a_j} = 1\\
\end{cases}\\
&=(C^{\ell-2}R_q)_{(a\setminus\{a_j\})}X_{(a_j)}\ket{x_1,\dots,x_n},
\end{align*}
which proves the equality for every basis state.

  For normality, \eqref{eq:commute} shows that conjugating a generator of
  $\gen{C^{\leq k}R_q}$ by $X_{(a_j)}$ returns an element of
  $\gen{C^{\leq k}R_q}$, since $|a\setminus\{a_j\}|=\ell-1\leq k$. Generators of
  $\gen{X}$ on qubits outside $a$ commute with $(C^{\ell-1}R_q)_{(a)}$ because they act on different qubits.
\end{proof}

\begin{proof}[Proof of \Cref{lem:split}]
  By \eqref{eq:standard-form} and \Cref{lem:commute} we may write
  $g=g'g''X_{(1)}^{x_1}g'''$ as in \eqref{eq:split}. The factors $g''$ and $g'''$
  involve no index set containing qubit $1$, so both act as $I\otimes(\cdot)$ and
  commute with the split. The factor $X_{(1)}^{x_1}$ exchanges the two cofactors
  when $x_1=1$ and does nothing when $x_1=0$. Finally $g'$ acts as the identity on
  $\ket0\otimes\co^{2^{n-1}}$, because every factor of $g'$ is a $C^{\ell-1}R_q$
  whose control set contains qubit $1$, and on $\ket1\otimes\co^{2^{n-1}}$ it acts
  as $\tilde g$ of \eqref{eq:gtilde}, by removing that control. Composing the three
  in the order of \eqref{eq:split} gives the displayed case distinction.

  For the independence claim, the exponent vector of $g$ decomposes as a direct sum
  over the index sets containing qubit $1$ and those not containing it, together
  with the vector $(x_1,\dots,x_n)$. The first block determines $\tilde g$ and is
  in bijection with $\z_q^{M_1}$, where $M_1=\sum_{\ell=1}^{k+1}\binom{n-1}{\ell-1}$
  is the number of exponent components of $\gen{C^{\leq k-1}R_q}$ on qubits
  $2,\dots,n$. The remaining blocks determine $g''g'''$. Since
  \eqref{eq:standard-form} places no constraint between the blocks, the two vary
  independently and each attains every value.
\end{proof}

\begin{proof}[Proof of \Cref{lem:arity}]
  Suppose $g$ maps one node of index $n$ to the other. By \Cref{lem:split}, $g$ acts on the
  cofactors as $g''g'''$ and as $\tilde gg''g'''$, both of which are elements of
  $\gen G$ on qubits $2,\dots,n$, because $G$ contains $C^{h}R_q$ for every
  $h\leq k$ and hence $\tilde g\in\gen{C^{\leq k-1}R_q}\subseteq\gen G$. So the
  cofactors are $\gen G$-equivalent, contradicting the assumption.
\end{proof}

\section{Omitted proofs of Section~\ref{sec:succinctness} (Succinctness)}\label{app:succ}

\begin{proof}[Proof of \Cref{thm:tower_is_hypergraphstate}]
  Induct on the number $n$ of levels. For $n=1$ the root node represents
  $\ket0+\ket1$ and the root edge is $2^{-1/2}$ or $2^{-1/2}(R_p)^d$ with
  $d\in\{1,\dots,p-1\}$, so it represents $\ket+$ or $(R_p)^d\ket+$. Both are
  hypergraph states, of the one-vertex hypergraph with $d$ copies of the edge
  $\{1\}$.

  For $n>1$, call the root $v_1$ and its child $v_2$. By the $0$-edge normal form
  the $0$-edge of $v_1$ carries $I^{\otimes n-1}$, its $1$-edge carries some
  $B\in\gen{C^{\leq k-2}R_p}$, and the root edge carries $2^{-n/2}R$ with
  $R\in\gen{C^{\leq k-2}R_p}$. So the diagram represents
  $2^{-n/2}R\ket{v_1}$ with $\ket{v_1}=\ket0\otimes\ket{v_2}+\ket1\otimes B\ket{v_2}$,
  in the unnormalized convention of \Cref{def:hypergraphstate}.

  By induction $\ket{v_2}=\ket{\hat\phi^{G}_p}$ for a $k$-hypergraph
  $G=(\{2,\dots,n\},E)$. Write
  $B=\prod_{1\leq\ell\leq k-1}\prod_{S\in\binom{[2,n]}{\ell}}[(C^{\ell-1}R_p)_{(S)}]^{d_S}$
  and let $F$ be the multiset containing $d_S$ copies of $\{1\}\cup S$ for every such
  $S$, so every element of $F$ has size at most $k$ and contains vertex $1$, while
  no element of $E$ does. Put $E'=E\uplus F$. Then
  \begin{align*}
    \ket{\hat\phi^{([n],E')}_p}
      &=\prod_{e\in F}(C^{|e|-1}R_p)_{(e)}\ \prod_{e\in E}(C^{|e|-1}R_p)_{(e)}\ (\ket0+\ket1)^{\otimes n}\\
      &=\prod_{e\in F}(C^{|e|-1}R_p)_{(e)}\ (\ket0+\ket1)\otimes\ket{\hat\phi^{G}_p}
       \tag{no $e\in E$ contains $1$}\\
      &=\ket0\otimes\ket{\hat\phi^{G}_p}
       +\ket1\otimes\prod_{e\in F}(C^{|e|-2}R_p)_{(e\setminus\{1\})}\ket{\hat\phi^{G}_p}
       \tag{every $e\in F$ contains $1$}\\
      &=\ket0\otimes\ket{v_2}+\ket1\otimes B\ket{v_2}=\ket{v_1} .
  \end{align*}
  Because $R$ is a product of matrices $C^{\ell}R_p$ with $\ell\leq k-2$, absorbing
  it adds further edges of size at most $k-1$ to $E'$, so $R\ket{v_1}$ is again the
  unnormalized state of a $k$-hypergraph. Multiplying by $2^{-n/2}$ gives
  \Cref{def:hypergraphstate} exactly. Running the induction on the unnormalized
  vectors is what keeps the normalization factor from being counted twice.
\end{proof}

\section{Omitted proofs of Section~\ref{sec:tractability} (Tractability)}\label{app:tract}

\subsection*{Measure and Sample}

Let $\ket\phi=label(e_r)\ket r$ for the root edge $e_r$ and root node $r$. The norm
recursion $\bra v\ket v=|\lambda_0|^2\bra{v_0}\ket{v_0}+|\lambda_1|^2\bra{v_1}\ket{v_1}$
follows from $\ket v=\ket0\otimes\ket{e_0}+\ket1\otimes\ket{e_1}$ and the fact that
every element of $\gen G$ is unitary, so only the scalar part of a label changes a
norm. One pass over the diagram computes it for all nodes.

For the amplitude, let $v$ have index $j$ and let $g$ be the LIM accumulated above
$v$. Split $g$ as in \eqref{eq:split}, and write $\tilde x_j$ for the exponent of
$X_{(j)}$ in $g$; it says whether $g$ flips qubit $j$, and it is not the queried bit
$x_j$. By \Cref{lem:split},
\begin{align*}
  \bra{x_j,\dots,x_n}g\ket v
  &=\bra{x_j}\bra{x_{j+1},\dots,x_n}
    \bigl(\ket0\otimes g''g'''\ket{e_{\tilde x_j}}+\ket1\otimes\tilde gg''g'''\ket{e_{\neg\tilde x_j}}\bigr)\\
  &=\begin{cases}
      \bra{x_{j+1},\dots,x_n}g''g'''\ket{e_{\tilde x_j}} & x_j=0,\\
      \bra{x_{j+1},\dots,x_n}\tilde gg''g'''\ket{e_{\neg\tilde x_j}} & x_j=1,
    \end{cases}
\end{align*}
so the walk enters the child $e_{x_j\oplus\tilde x_j}$ with the accumulated LIM
updated by one product. Each edge label is $label(e_i)$ times the node, so the
recursion is a sequence of $n$ LIM products, each $O(M)$ additions in $\z_q$ by
\eqref{eq:iso}. Unfolding it, $\bra x\ket\phi$ is the product of the scalar parts of
the labels along the path selected by the bits $x_j\oplus\tilde x_j$.

Sampling walks the same path. At a node $v$ of index $j$ with accumulated LIM $g$,
toss a coin that lands heads with probability
$\bra{e_{\tilde x_j}}\ket{e_{\tilde x_j}}/\bra v\ket v$, set $x_j=0$ on heads and
$x_j=1$ otherwise, and continue into $e_{\tilde x_j}$ or $e_{\neg\tilde x_j}$
respectively, updating $g$ as above. The probability of producing $x$ telescopes to
$\bra{e^{(n+1)}}\ket{e^{(n+1)}}/\bra{v^{(1)}}\ket{v^{(1)}}$ times the product of the
squared scalar parts along the path, where $v^{(1)}$ is the root. Since
$\bra\phi\ket\phi=|\lambda_{e_r}|^2\bra{v^{(1)}}\ket{v^{(1)}}$, that equals
$|\bra x\ket\phi|^2/\bra\phi\ket\phi$ by the previous paragraph. Both queries run in
time linear in the diagram and in $M$.

\subsection*{The six instances of Lemma~\ref{lem:transfer}}

We verify \eqref{eq:recursive} for the three gate families and the two kinds of
group. Throughout $v$ has index $b$ and
$\ket v=\ket0\otimes L_0\ket{v_0}+\ket1\otimes L_1\ket{v_1}$.

\paragraph*{$X_{(a)}$ on a diagonal group.} Take
$\mathcal V_b=\{X_{(a)}\}$ for $b\geq a$ and $\{I\}$ below. For $b>a$ the gate
commutes past the label up to \eqref{eq:R_q--X__commutation}, giving
$X_{(a)}\ket v=\ket0\otimes L_0'(X_{(a)}\ket{v_0})+\ket1\otimes L_1'(X_{(a)}\ket{v_1})$
with $L_i'$ obtained from $L_i$ by replacing its factor $(R_q)^{d}$ on qubit $a$ by
$\omega_q^{d}(R_q)^{-d}$. For $b=a$ the gate exchanges the two children:
$X_{(a)}\ket v=\ket0\otimes L_1\ket{v_1}+\ket1\otimes L_0\ket{v_0}$, and both
children get $I\in\mathcal V_{b-1}$.

\paragraph*{$X_{(a)}$ on a group containing $X$.} Identical, except that $L_i$ may
itself contain bit flips; these commute with $X_{(a)}$, and the diagonal part is
updated by \eqref{eq:commute} exactly as above. In both cases
$|\mathcal V_b|\leq1$.

\paragraph*{$(R_m)^d_{(a)}$.} Take $\mathcal V_b=\{(R_m)^d_{(a)}\}$ for $b>a$,
$\{I\}$ for $b<a$, and at $b=a$ note that $(R_m)^d$ acts only on the $1$-branch:
$(R_m)^d_{(a)}\ket v=\ket0\otimes L_0\ket{v_0}+\ket1\otimes\omega_m^{d}L_1\ket{v_1}$,
so both children get $I$. For $b>a$ on a diagonal group the gate commutes with
$L_i$ outright. On a group containing $X$, pushing $(R_m)^d$ past a label whose
$X$-part flips qubit $a$ returns $(R_m)^{-d}$ by
\eqref{eq:R_q--X__commutation}, so $\mathcal V_b=\{(R_m)^{d}_{(a)},(R_m)^{-d}_{(a)}\}$
and $|\mathcal V_b|\leq2$.

\paragraph*{$(C^{k'}R_m)^d$ on qubits $a_1<\dots<a_{k'+1}$.}
Take $\mathcal V_b$ as in \Cref{lem:transfer}, that is,
$\{(C^{k'}R_m)^{\pm d}\}$ above the highest control,
$\{I,(C^{k''}R_m)^{\pm d}\}$ between controls $a_{k''}$ and $a_{k''+1}$, and
$\{I\}$ below the lowest.
Off a control level, $b\notin\{a_1,\dots,a_{k'+1}\}$, the gate is diagonal and
commutes with a diagonal $L_i$, so \eqref{eq:recursive} holds with $L_i'=L_i$ and
the same gate passed down. At a control level $b=a_{k''+1}$ the gate acts as the
identity on the $0$-branch and loses its highest control on the $1$-branch, by the
definition of $C^{k''}R_m$:
\begin{equation*}
  (C^{k''}R_m)^{d}_{(a_1,\dots,a_{k''+1})}\ket v
   =\ket0\otimes L_0\ket{v_0}+\ket1\otimes L_1\bigl((C^{k''-1}R_m)^{d}_{(a_1,\dots,a_{k''})}\ket{v_1}\bigr),
\end{equation*}
using once more that the two are diagonal and commute. Below the lowest control the
gate is the identity. For a group containing $X$, a label whose $X$-part flips one
of the controls turns $d$ into $-d$ by \eqref{eq:commute}, so
$\mathcal V_b$ gains $(C^{k''}R_m)^{-d}$ and $|\mathcal V_b|\leq3$. This is the
only case with three elements, which proves the bound of \Cref{lem:transfer}. The
gates $Z=R_2$, $T=R_8$ and $C^{k'}R_m$ of \Cref{tab:tractability} are all instances.
\qed

\subsection*{Proof of Theorem~\ref{thm:swap}}

Linear-size $\gen I$-LIMDDs for $\ket+^{\otimes n}$ and $\ket{rot^n}$ are given by
a direct generalization of \cite[Lemma 22]{vinkhuijzen2024a}: both are
product states, so each level has one node. Hence $\ket{sum^n}$ has $n+3$ nodes,
and enlarging $\gen G$ never enlarges a diagram.

After the swap,
$\textit{SWAP}_{1,n+2}\ket{sum^n}=\ket0\otimes\ket+^{\otimes n}\otimes\ket0+\ket0\otimes\ket{rot^n}\otimes\ket1$.
Fix $x_1=0$ and $x_2\dots x_n=a_2\dots a_n$. The branch $x_{n+2}=0$ is the constant
$1$. The branch $x_{n+2}=1$ is $\prod_{j=1}^{n}e^{\pi i2^{-j}x_{j+1}/q}$, whose
factors with $j\leq n-1$ are fixed by $a$ and whose factor $j=n$ depends on
$x_{n+1}$. That is exactly $f_a$ of \eqref{eq:fa}, one function for each
$a\in\{0,1\}^n$ with $a_1=0$.

Let $v_a,v_b$ be the corresponding nodes and suppose $\ket{v_a}=\alpha g\ket{v_b}$
with $\alpha\in\co$ and $g\in\gen{C^{\leq k}R_q,X}$. The node has index $2$, so
every LIM acting on it acts on two qubits and therefore carries at most one control,
whatever $k$ is. By \eqref{eq:standard-form},
\begin{equation*}
  g=(R_q)^{d_{n+1}}_{(n+1)}(R_q)^{d_{n+2}}_{(n+2)}(CR_q)^{d_{n+1,n+2}}_{(n+1,n+2)}
    X_{(n+1)}^{e_{n+1}}X_{(n+2)}^{e_{n+2}},
\end{equation*}
where $d_{n+1,n+2}$ is a single exponent and $e_{n+1},e_{n+2}\in\{0,1\}$ are the
bit-flip exponents, not the variables. From $\bra{00}\ket{v_a}=1$ we get
$\alpha=(\bra{00}g\ket{v_b})^{-1}$. Put
$S:=\sum_{j=1}^{n-1}(a_{j+1}-b_{j+1})2^{-j}$, so $|S|<1$, and note that $S$ is an
integer only when $S=0$, which forces $a=b$.

Comparing $\bra{01}$ gives, for $e_{n+2}=0$,
\begin{equation*}
  \omega_q^{d_{n+2}}=e^{\pi i\,(S-e_{n+1}2^{-n})/q},
  \qquad\text{i.e.}\qquad
  2d_{n+2}-\bigl(S-e_{n+1}2^{-n}\bigr)\in 2q\z .
\end{equation*}
The left side is an integer minus $S-e_{n+1}2^{-n}$, so $S-e_{n+1}2^{-n}\in\z$. If
$e_{n+1}=0$ this gives $S\in\z$, hence $S=0$, hence $a=b$. If $e_{n+1}=1$ then
$S-2^{-n}$ has $2$-adic valuation exactly $-n<0$, because $S$ is a dyadic rational
of valuation at least $-(n-1)$; so it is never an integer, and this case is
impossible. Comparing $\bra{10}$ instead, which gives $1=\bra{10}\ket{v_a}$,
disposes of $e_{n+2}=1$: the same computation yields
$2d_{n+1}\mp2^{-n}\in2q\z$, and $\pm2^{-n}$ is never an integer for $n\geq1$, so
$d_{n+1}$ does not exist. The four patterns of $(e_{n+1},e_{n+2})$ are exhaustive,
and each is ruled out, so distinct $a\neq b$ give inequivalent nodes.

There are $2^{n-1}$ such $a$, so the second-to-last layer has $2^{n-1}$ pairwise
inequivalent nodes, none of which can be merged by \Cref{lem:arity}. Hence the
diagram is exponential and $\textit{SWAP}$ is not polytime. Finally
$\textit{SWAP}_{(a,b)}=CX_{(a,b)}CX_{(b,a)}CX_{(a,b)}$ makes $CX$ intractable, and
$CX_{(a,b)}=H_{(b)}CZ_{(a,b)}H_{(b)}$ with $CZ$ tractable by \Cref{lem:transfer}
makes $H$ intractable.
\qed

\section{Proof of Theorem~\ref{thm:canonicity}}\label{app:canon}

\thmcanonicity*
\begin{proof}
  We argue for $\gen{C^{\leq k}R_q,X}$; dropping the cases with $x_1=1$ gives the
  diagonal case. Induction on $n$. For $n=0$ the vector is a scalar $\lambda$,
  represented by an edge labeled $\lambda$ into the leaf, which is unique because
  the leaf is unique and denotes $1$. Let $n\geq1$ and split
  $\ket\phi=\ket0\otimes\ket{\phi_0}+\ket1\otimes\ket{\phi_1}$.

  \emph{Existence.} By induction there are nodes $v_0,v_1$ and $g_0,g_1\in\gen G$
  with $\ket{\phi_i}=g_i\ket{v_i}$. If $v_0\preccurlyeq v_1$ then
  $\ket\phi=(I\otimes g_0)\left(\ket0\otimes\ket{v_0}+\ket1\otimes g_0^{-1}g_1\ket{v_1}\right)$,
  and R4 turns $g_0^{-1}g_1$ into the canonical high label. Otherwise
  $\ket\phi=(X\otimes g_1)\left(\ket0\otimes\ket{v_1}+\ket1\otimes g_1^{-1}g_0\ket{v_0}\right)$
  and the same applies. The resulting node satisfies R1--R5 by construction.

  \emph{Uniqueness.} Let $g_v\ket v=\ket\phi=g_w\ket w$ with $v,w$ both satisfying
  R1--R5, and write $\ket v=\ket0\otimes\ket{v_0}+\ket1\otimes B_v\ket{v_1}$ and
  $\ket w=\ket0\otimes\ket{w_0}+\ket1\otimes B_w\ket{w_1}$. Apply \Cref{lem:split}
  to $g_w^{-1}g_v$. If $x_1=0$ then $g''g'''\ket{v_0}=\ket{w_0}$ and
  $\tilde gg''g'''B_v\ket{v_1}=B_w\ket{w_1}$, so $v_0,w_0$ are equivalent and so
  are $v_1,w_1$; by induction $v_0=w_0$ and $v_1=w_1$. If $x_1=1$ then instead
  $v_1=w_0$ and $v_0=w_1$, so R2 gives
  $v_0\preccurlyeq v_1=w_0\preccurlyeq w_1=v_0$, and antisymmetry of
  $\preccurlyeq$ forces $v_0=v_1$, so again all four children coincide. In both
  cases $B_v$ and $B_w$ lie in the same set \eqref{eq:true_label_class}, the second
  case using the inverse-twisted coset of \Cref{prop:r4class}. Rule R4 selects the
  same element of that set for both, so $B_v=B_w$. Hence $v$ and $w$ have the same
  children and the same labels, and R5 merges them.
\end{proof}

\section{Group theory of the bit-flip extension}\label{app:group}

In \Cref{sec:machinery} we described the group $\gen{C^{\leq k}R_q}$ and its equivalence to $\z_q^M$. We now describe the group $\gen{C^{\leq k}R_q,X}$ and its equivalence to the semidirect product group $\z_q^{M+1}\rtimes_\varphi\z_2^{\,n}$.

By \eqref{eq:commute} with $\ell=1$ we get
$X_{(j)}(R_q)_{(j)}X_{(j)}=\omega_q\left[(R_q)_{(j)}\right]^{q-1}$, so the global
phase $\omega_q\cdot I$ lies in $\gen{C^{\leq k}R_q,X}$ although it does not lie in
$\gen{C^{\leq k}R_q}$. Recording it as the exponent $d_\emptyset$ of the empty index
set, the diagonal part of $\gen{C^{\leq k}R_q,X}$ is $\z_q^{M+1}$. We know that $\gen{C^{\leq k}R_q}$ is normal in $\gen{C^{\leq k}R_q,X}$.

We have that $\gen{C^{\leq k}R_q,X}$ is thus a semi-direct product:
\begin{equation}\label{eq:isoX}
  \gen{C^{\leq k}R_q,X}\;\cong\;\gen{C^{\leq k}R_q}\rtimes_\varphi\gen{X}\;\cong\;\z_q^{M+1}\rtimes_\varphi\z_2^{\,n} .
\end{equation}
The action is the conjugation $\varphi_h(g)=hgh^{-1}$ with $g\in\gen{C^{\leq k}R_q}, h\in\gen{X}$, which, according to \Cref{lem:commute}, evaluates as
$\varphi_{X_{(j)}}\bigl((C^{\ell}R_q)_{(S)}\bigr)
 =\bigl[(C^{\ell}R_q)_{(S)}\bigr]^{-1}(C^{\ell-1}R_q)_{(S\setminus\{j\})}$ when
$j\in S$, and equals the identity when $j\notin S$. Multiplication is
$(a,b)(a',b')=(a\,\varphi_b(a'),bb')$, for $a,a'\in\gen{C^{\leq k}R_q}$ and $b,b'\in\gen{X}$.

We write
\begin{equation*}
  \pi_a(g):=d_a\in\z_q
\end{equation*}
for the \emph{projection} of $g$ onto the component of the index set $a$
Each $\pi_a$ is a surjective group
homomorphism onto $\z_q$, and each $\varphi_r$ is $\z_q$-linear on exponent vectors $g$. Both are computable with $O(M)$ operations in $\z_q$ by \Cref{lem:commute}.

\subsection*{Subgroups}

By \cite[Theorem 1.3.3]{usenko1991subgroups}, every subgroup
$U\leq\gen{C^{\leq k}R_q,X}$ is determined by a triple $(L,R,\theta)$ with
$L\leq\gen{C^{\leq k}R_q}$, $R\leq\gen X$ and a normal crossed homomorphism
$\theta:R\to\gen{C^{\leq k}R_q}/L$, through
\begin{equation}\label{eq:LRtheta}
  U=\{(u\cdot\theta(h),h)\mid u\in L,\ h\in R\} .
\end{equation}
Being a normal crossed homomorphism means
$\theta(gh)=\theta(g)\cdot\varphi_g(\theta(h))\bmod L$, so $\theta$ is determined by
its values on a generating set of $R$. We represent $L$ and $R$ by matrices in
Howell normal form over $\z_q$ and $\z_2$ respectively, and $\theta$ by the matrix
of images of the generators of $R$. All three are $O(M)$ wide, so a subgroup is
stored in $O(M^2)$ digits.

Finally, note that stabilizers are submodules and not merely subgroups. Under
\eqref{eq:iso} the condition $g\ket v=\ket v$ is $\z_q$-linear in the exponent
vector of $g$, so $L=Stab(v)\cap\gen{C^{\leq k}R_q}$ is a $\z_q$-submodule of
$\z_q^M$ for every $q$, prime or not. This is what makes the Howell form the right
representation. Note that for $q$ prime, the submodule is actually a subgroup as $\z_q$ is a field.

\section{The canonicity algorithm with the bit flip}\label{app:algoX}

Throughout this appendix $\gen G=\gen{C^{\leq k}R_q,X}$, for an arbitrary
$q\geq1$. Rules R1--R3 and R5 are applied exactly as in the proof of
\Cref{thm:polytime}; we give R4, in two steps: first compute $Stab(v)$, then
minimize the high edge label. 

We represent stabilizers of nodes, which are subgroups of $\gen{C^{\leq k}R_q,X}$, by the representation of \Cref{app:group}. 
Such subgroups are therefore represented by a subgroup $L$ of $\gen{C^{\leq k}R_q}$ in Howell normal form, a subgroup $R$ of $\gen{X}$ in reduced row echelon form, and a normal crossed homomorphism $\theta$.
We store only the images of $\theta$ under the generators of $R$, as they fully determine $\theta$ as it is a normal crossed homomorphism.

Subgroups are therefore triples $(L,R,\theta)$ as in \eqref{eq:LRtheta}. Let
$\ket v=\ket0\otimes\ket{v_0}+\ket1\otimes\ket{v_1'}$ with $\ket{v_1'}=E\ket{v_1}$,
and let $Stab(v_0)=(L_0,R_0,\theta_0)$ and $Stab(v_1)=(L_1,R_1,\theta_1)$ be known
from the recursive calls. We first show how to compute the stabilizers of $v_1'$.

\subsection*{Conjugating a stabilizer along an edge}

\begin{lemma}\label{lem:conjstab}
  Let $E=(\ell,r)\in\gen{C^{\leq k}R_q,X}$ in the presentation of \Cref{app:group}. Then
  $Stab(v_1')=E\,Stab(v_1)\,E^{-1}=(L_1',R_1',\theta_1')$ with
  \begin{equation*}
    R_1'=R_1,\qquad
    L_1'=\varphi_r(L_1),\qquad
    \theta_1'(b)=\varphi_r(\theta_1(b))\cdot\ell\,\varphi_b(\ell)^{-1}\bmod L_1' ,
  \end{equation*}
  and all three are computable with $O(M^{3})$ operations in $\z_q$.
\end{lemma}
\begin{proof}
  $\ket{v_1'}=E\ket{v_1}$ gives $Stab(v_1')=E\,Stab(v_1)\,E^{-1}$ directly. Since
  $\gen{C^{\leq k}R_q}$ is abelian and $r^2=I$, conjugation by $(\ell,r)$ sends
  $(a,b)$ to $(\varphi_r(a)\cdot\ell\,\varphi_b(\ell)^{-1},\,b)$, which is the
  displayed triple. Each $\varphi$ is $\z_q$-linear on exponent vectors by
  \Cref{lem:commute}, so each item is a matrix product followed by one Howell
  reduction.
\end{proof}

\subsection*{Computing \texorpdfstring{$Stab(v)$}{Stab(v)}}
Now we show how to compute the stabilizers of $v$.

Let $g\in Stab(v)$ and split it as $g=\ctl(\tilde g)g''X_{(1)}^{x_1}g'''$. By
\Cref{lem:split},
\begin{equation*}
  \ket v=\begin{cases}
    \ket0\otimes g''g'''\ket{v_0}+\ket1\otimes\tilde gg''g'''\ket{v_1'} & x_1=0,\\
    \ket0\otimes g''g'''\ket{v_1'}+\ket1\otimes\tilde gg''g'''\ket{v_0} & x_1=1.
  \end{cases}
\end{equation*}
For $x_1=1$ this forces $\ket{v_0}$ and $\ket{v_1'}$, hence $\ket{v_0}$ and
$\ket{v_1}$, to be $\gen G$-equivalent. Since the children are already canonical,
that means $v_0=v_1$. So if $v_0\neq v_1$ there is no stabilizer with $x_1=1$. For
$x_1=0$ we get $g''g'''\in Stab(v_0)$ and $\tilde gg''g'''\in Stab(v_1')$.

\paragraph*{Compute $L$.} Since $\theta(I)=I$, the subgroup $L$ consists exactly of
the stabilizers with no $\gen X$-part, that is $x_1=0$ and $g'''=I$. The condition
above then reads $g''\in L_0$ and $\tilde gg''\in L_1'$, which is equal to the case without bit-flip as in
\eqref{eq:stabcomap}: we only need to replace $Stab(v),Stab(v_0),Stab(v_1)$ by $L,L_0,L_1'$ respectively. 
As in the proof of \Cref{thm:polytime}, this shows that $L$ is computable in polynomial time.

\paragraph*{Compute $R$.} First the elements with $x_1=0$. From
$g''g'''\in Stab(v_0)$ and $\tilde gg''g'''\in Stab(v_1')$, writing $r:=g'''$, we get
$r\in R_0\cap R_1'$ and, dropping the $\gen X$-parts,
$g''=\ell_0\theta_0(r)$ and $\tilde gg''=\ell_1\theta_1'(r)$ for some
$\ell_0\in L_0$ and $\ell_1\in L'_1$. Combining,
\begin{equation}\label{eq:Rcond}
  \theta'(r)\,\ell'\,\tilde g=I,
  \qquad
  \theta'(r):=\theta_0(r)\theta_1'(r)^{-1},\quad \ell'\in L_0\cdot(L_1')^{-1},
  \quad\tilde g\in\gen{C^{\leq k-1}R_q}.
\end{equation}
Because $\tilde g$ is free in $\gen{C^{\leq k-1}R_q}$, which is exactly the set of
elements supported on the components of arity at most $k$, \eqref{eq:Rcond} is
equivalent to $\pi_a(\theta'(r)\ell')=0$ for every $a$ with $|a|=k+1$. Fix such an
$a$. The map $\chi_a(r):=|a\cap r|\bmod2$ is $\z_2$-linear on $R_0\cap R_1'$, and
$\pi_a\circ\varphi_r$ is the identity when $\chi_a(r)=0$ and inversion otherwise. So
$\pi_a\circ\theta'$ is $\z_q$-linear on $\ker\chi_a$. Choose a basis accordingly:
if $\chi_a\equiv0$ take any basis of $R_0\cap R_1'$; otherwise $\ker\chi_a$ has
index $2$, so take a basis of $\ker\chi_a$ and add one element $r_m$ with
$\chi_a(r_m)=1$. Gaussian elimination over $\z_2$ finds this basis in $O(n^3)$.
With that basis $\pi_a(\theta'(r)\ell')$ is affine in $(\ell',r)$ on each of the two
cosets of $\ker\chi_a$, and its kernel is computed by linear algebra over $\z_q$.
Intersecting over all $a$ of arity $k+1$ and taking $\gen X$-parts gives
\begin{equation*}
  R\cap\gen{X_{(2)},\dots,X_{(n)}}
    =\gen X\cap\bigcap_{|a|=k+1}\ker\bigl(\pi_a(\theta'(\cdot)\,\cdot)\bigr).
\end{equation*}

Now the elements with $x_1=1$. As shown above these exist only when $v_0=v_1$, so
assume that. One such element suffices, since the product of two of them has
$x_1=0$ and is already found. If $E=I$ then $X_{(1)}\in Stab(v)$ and we add it to
$R$. In general, write $u:=g''g'''$ for the part acting on qubits $2,\dots,n$. Since
$X_{(1)}$ and $g'''$ have disjoint support they commute, so the $x_1=1$ split reads
$g=\ctl(\tilde g)\,u\,X_{(1)}$. With $v_0=v_1$ and $\ket{v_1'}=E\ket{v_0}$, the two
conditions of the case distinction above, $u\ket{v_1'}=\ket{v_0}$ and
$\tilde gu\ket{v_0}=\ket{v_1'}$, become
\begin{equation}\label{eq:x1cond}
  uE\in Stab(v_0)
  \qquad\text{and}\qquad
  E^{-1}\tilde gu\in Stab(v_0).
\end{equation}
Setting $s:=uE$, so that $u=sE^{-1}$, the pair \eqref{eq:x1cond} says exactly that
there are $s\in Stab(v_0)$ and $\tilde g\in\gen{C^{\leq k-1}R_q}$ with
$E^{-1}\tilde g\,s\,E^{-1}\in Stab(v_0)$. Conversely, any such $s$ and $\tilde g$
yield $u:=sE^{-1}$ satisfying \eqref{eq:x1cond}, and this $u$ is admissible because
by \Cref{lem:split} the factor $g''g'''$ ranges over the whole group on qubits
$2,\dots,n$ freely; the stabilizer obtained is
\begin{equation*}
  g=\ctl(\tilde g)\,s\,E^{-1}X_{(1)} .
\end{equation*}
That is one coset-membership test in the subgroup generated by $Stab(v_0)$ and
$\gen{C^{\leq k-1}R_q}$, which is a Howell reduction and a residue test, in time
polynomial in $M$.

\paragraph*{Compute $\theta$.} The map $\theta$ needs no separate search; it is read
off the witnesses of the solve that produced $R$. Run the linear solve of
\eqref{eq:Rcond} so that for each generator $r$ of $R$ it returns a witness pair
$(\ell',\tilde g)$ realizing the solution, which Gaussian elimination provides at no
extra cost. Recover the corresponding $g''=\ell_0\theta_0(r)$ and set
\begin{equation*}
  \theta(r):=\ctl(\tilde g)\cdot g'' .
\end{equation*}
By construction $(\theta(r),r)\in Stab(v)$, so $\theta(r)$ is a valid representative,
and two witnesses for the same $r$ differ by an element of $L$, so $\theta(r)$ is
well defined in $\gen{C^{\leq k}R_q}/L$. For the generator $X_{(1)}$, when present,
take the witness produced by the membership test above. Extending $\theta$ from
generators to all of $R$ needs no computation: the crossed-homomorphism law
$\theta(gh)=\theta(g)\varphi_g(\theta(h))$ determines it for every $q$; at $q=2$ the
twist is trivial and the law degenerates to linearity modulo
$\gen{C^{\leq k-1}R_q}$, but nothing below needs that degeneration. This completes
$(L,R,\theta)$.

\subsection*{Minimizing the high edge label}

\Cref{sec:algo} states this minimization as five steps; what follows derives them.
Steps 0 and 2 there are the presentations, Step 1 is the $\gen X$-component
minimized below, Step 3 is the loop over top-arity components, and Step 4 is the
scalar.%

When $v_0=v_1$ the class is the union of the two double cosets of
\Cref{prop:r4class}, so the whole minimization below is run twice --- once from the
base point $E$ and once from the inverse-twisted base point $E^{-1}$ --- and the
smaller of the two answers is taken. That doubles the computational cost and changes nothing
else, so we describe one run.

By \Cref{prop:r4class} we minimize $g_0^{-1}Eg_1$ over $g_i\in Stab(v_i)$ together
with the free factor $\gen{C^{\leq k-1}R_q}$. Writing $E=(\ell,r)$ and
$g_i=(\ell_i\theta_i(r_i),r_i)$ and using
$(a,b)^{-1}(a',b')=(\varphi_b(a^{-1}a'),bb')$ repeatedly,
\begin{equation*}
  g_0^{-1}Eg_1=\bigl(\varphi_{r_0}\bigl(\ell_0^{-1}\theta_0(r_0)^{-1}\ell\,\varphi_r(\ell_1\theta_1(r_1))\bigr),\ rr_0r_1\bigr).
\end{equation*}
Minimize the $\gen X$-component $rr_0r_1$ first. Its minimum is attained on
$r_0=\tilde r_0r'$, $r_1=\tilde r_1r'$ for fixed $\tilde r_0,\tilde r_1$ and
$r'\in R_0\cap R_1$ free. Substituting and collecting the terms that do not depend
on $r'$ into a constant $A$, the objective becomes
\begin{equation}\label{eq:objective}
  \bigl(\varphi_{r'}\bigl[A\cdot\ell'\cdot\theta'(r')\bigr],\ r\tilde r_0\tilde r_1\bigr),
  \qquad
  \ell'\in L':=\gen{\varphi_{\tilde r_0}(L_0),\varphi_{\tilde r_0r}(L_1)},
\end{equation}
with $\theta'(r')=\varphi_{r\tilde r_0\tilde r_1}(\theta_1(r'))\theta_0(r')^{-1}$,
a normal crossed homomorphism computable from the stored representations of $\theta_0,\theta_1$.

Note that only the components
$a$ with $|a|=k+1$ need to be minimized; every component of smaller arity is free by the free factor $\gen{C^{\leq k-1}R_q}$,
and can hence be set to $0$. The components of arity $k+1$ are minimized one at a
time, always in the fixed order \eqref{eq:order}, each one against the domain the
previous one left behind; we describe one pass of that loop. Fix such an $a$ of
arity $k+1$. By \Cref{lem:commute} the map
$\pi_a\circ\varphi_{r'}$ is multiplication by $\chi_a(r')=(-1)^{|a\cap r'|}$ --- the
identity when $|a\cap r'|$ is even and inversion when it is odd; we call it a
\emph{sign} for that reason, and note that all we use is that $\pm1$ are
\emph{units} of $\z_q$. Note that $\varphi_{r'}$ is part of the objective: it is
\begin{equation}\label{eq:obj}
  \pi_a\bigl[\varphi_{r'}(A\,\ell'\,\theta'(r'))\bigr]
  =\chi_a(r')\Bigl(\pi_a[A]+\pi_a[\ell'\theta'(r')]\Bigr) ,
\end{equation}
which is \emph{affine}, not linear, in $(\ell',r')$ --- the constant $\pi_a[A]$ does
not vanish --- and whose only nonlinearity is the sign $\chi_a(r')$.

Minimizing \eqref{eq:obj} is one computation. Choose a basis of
$L'\times(R_0\cap R_1)$ in which at most one $\gen X$-generator $r_m$ has
$\chi_a(r_m)=-1$; Gaussian elimination over $\z_2$ finds it in $O(n^3)$. On the even
coset $\pi_a[\ell'\theta'(\cdot)]$ is a homomorphism into $\z_q$, so its image is a
subgroup $H_a\leq\z_q$, necessarily cyclic; on the odd coset the image is the single
translate $\pi_a[\ell'_m\theta'(r_m)]+H_a$. So the achievable values of
\eqref{eq:obj} are at most two cosets of one cyclic subgroup of $\z_q$, and the
least of them, $A_{(a)}$, is a constant number of residue computations. This is the
only place the union of \Cref{sec:algo} appears, and there it is free.

The new minimization domain is a single coset for every $q$. It is
\begin{equation}\label{eq:dom}
  \Bigl\{(\ell',r')\;\Bigm|\;
    \pi_a\bigl[\varphi_{r'}(A\,\ell'\,\theta'(r'))\bigr]=\pi_a[A]\Bigr\} ,
\end{equation}
translated by any one minimizer. By \eqref{eq:obj}, membership in \eqref{eq:dom}
reads $\pi_a[\ell'\theta'(r')]=0$ when $\chi_a(r')=1$ and
$\pi_a[\ell'\theta'(r')]=-2\pi_a[A]$ when $\chi_a(r')=-1$ --- there the sign flips the
constant as well, so $-\bigl(\pi_a[A]+\pi_a[\ell'\theta'(r')]\bigr)=\pi_a[A]$; the two branches are the
two halves of one condition, and \eqref{eq:dom} is a subgroup of
$L'\times(R_0\cap R_1)$ because it is the kernel of a crossed homomorphism ---
which is a subgroup whether or not the twist is trivial, and whether or not the
\emph{image} is closed. Each branch is one linear solve over $\z_q$, and their union
is presented by the solved form of the even branch together with one extra
generator from the odd branch when that branch is nonempty. Replace $A$ by
$A_{(a)}$ and the domain by that coset, and continue with the next $a$ in the
decreasing-arity order \eqref{eq:order}. After the $\binom{n}{k+1}$ components of
top arity the label is determined, and all remaining components are $0$.

Every step is a Howell reduction or a linear solve over $\z_q$ on matrices of width
$O(M)=O(n^{k+1})$, so R4 runs in time polynomial in $M$ per stored node --- for
every $q$. Together with R1--R3 and R5 this gives the second half of
\Cref{thm:polytime}.

\end{document}